\documentclass[11pt]{article}

\usepackage{amsmath,amsfonts,mathrsfs,latexsym,color,array}

\usepackage{fancyhdr}
\usepackage{cite}
\usepackage{url}
\usepackage{ntheorem}
\usepackage{graphicx}
\usepackage{epstopdf}

\newcommand{\ptcxx}[1]{{\color{red}\mnote{{\color{red}{\bf ptc:}
#1} }}}

\renewcommand{\ptcxx}[1]{}

\newcommand{\cg}{{\check{g}}}

\newcommand{\cphi}{{\check{\phi}}}
\newcommand{\cR}{\check{R}}
\newcommand{\cbeta}{\check{\beta}}
\newcommand{\csigma}{\check{\sigma}}

\newcommand{\e}{{\varepsilon}}

\newcommand{\tL}{\tilde L}%
\newcommand{\tR}{\tilde R}%
\newcommand{\calO}{\mathcal O}
\newcommand{\calR}{\mathcal R}
\newcommand{\calL}{\mathcal L}
\newcommand{\calU}{\mathcal U}
\newcommand{\calV}{\mathcal V}

\definecolor{bluem}{rgb}{0,0,0.5}

\definecolor{mycolor}{cmyk}{0.5,0.1,0.5,0}
\definecolor{michel}{rgb}{0.5,0.9,0.9}
\definecolor{turquoise}{rgb}{0.25,0.8,0.7}
\definecolor{bluem}{rgb}{0,0,0.5}

\definecolor{MDB}{rgb}{0,0.08,0.45}
\definecolor{MyDarkBlue}{rgb}{0,0.08,0.45}

\definecolor{MLM}{cmyk}{0.1,0.8,0,0.1}
\definecolor{MyLightMagenta}{cmyk}{0.1,0.8,0,0.1}

\definecolor{HP}{rgb}{1,0.09,0.58}

\newcommand{\za}{{\mathring a}}
\newcommand{\zu}{{\mathring u}}
\newcommand{\ztR}{{\mathring{\tR}}}
\newcommand{\ztg}{{\mathring{\tg}}}
\newcommand{\zR}{{\mathring R}}
\newcommand{\zsigma}{{\mathring \tsigma }}
\newcommand{\zbeta}{{\mathring \beta }}
\newcommand{\zw}{{\mathring w}}
\newcommand{\zphi}{{\mathring \phi}}

\newcommand{\zhR}{{\mathring \hR}}

{\catcode `\@=11 \global\let\AddToReset=\@addtoreset}
\AddToReset{equation}{section}

{\catcode `\@=11 \global\let\AddToReset=\@addtoreset}
\AddToReset{figure}{section}

 \newcommand{\tsigma}{{\tilde \sigma}}
 
 \newcommand{\zhsigma}{{\mathring \hsigma}}
 \newcommand{\ztsigma}{{\mathring \tsigma}}
 \newcommand{\hsigma}{{\hat \sigma}}
 \newcommand{\hR}{{\hat R}}
 \newcommand{\hg}{{\hat g}}

\newcommand{\zf}{\mathring f}

\newtheorem{Theorem} {\sc  Theorem\rm} [section]

\newtheorem{Lemma} [Theorem] {\sc  Lemma\rm}
\newtheorem{lemma} [Theorem] {\sc  Lemma\rm}
\newtheorem{proposition} [Theorem] {\sc  Proposition\rm}
\newtheorem{Proposition} [Theorem] {\sc  Proposition\rm}

\theorembodyfont{\upshape}

\newcommand{\beqar}{\begin{deqarr}}
\newcommand{\eeqar}{\end{deqarr}}

\newcommand{\beaa}{\begin{eqnarray*}}
\newcommand{\eeaa}{\end{eqnarray*}}

\newcommand{\bel}[1]{\begin{equation}\label{#1}}
\newcommand{\bea}{\begin{eqnarray}}
\newcommand{\bean}{\begin{eqnarray}\nonumber}
\newcommand{\beal}[1]{\begin{eqnarray}\label{#1}}
\newcommand{\eea}{\end{eqnarray}}
\newcommand{\eeal}[1]{\label{#1}\end{eqnarray}}

\newcommand{\qed}{\hfill $\Box$ \medskip}
\newcommand{\proof}{\noindent {\sc Proof:\ }}
\newcommand{\be}{\begin{equation}}
\newcommand{\eeq}{\end{equation}}
\newcommand{\ee}{\end{equation}}
\newcommand{\beqa}{\begin{eqnarray}}
\newcommand{\eeqa}{\end{eqnarray}}
\newcommand{\beqan}{\begin{eqnarray*}}
\newcommand{\eeqan}{\end{eqnarray*}}
\newcommand{\ba}{\begin{array}}
\newcommand{\ea}{\end{array}}

\newcommand{\cC}{{\cal C}}
\newcommand{\const}{\mbox{\rm const}} %constants
\newcommand{\calC}{\mathcal C}

\DeclareFontFamily{OT1}{rsfs}{} \DeclareFontShape{OT1}{rsfs}{m}{n}{
<-7> rsfs5 <7-10> rsfs7 <10-> rsfs10}{}
\DeclareMathAlphabet{\mycal}{OT1}{rsfs}{m}{n}

\newcommand{\R}{\mathbb R}
\newcommand{\RR}{\mathbb R}

\newcommand{\bit}{\begin{itemize}}
\newcommand{\eit}{\end{itemize}}

\newcommand{\del}{\partial}

\newcommand{\tg}{{\tilde g}}
\newcommand{\tphi}{{\tilde \phi}}
\newcommand{\dtphi}{\underline{{\tilde \phi}}}
\newcommand{\utphi}{\overline{{\tilde \phi}}}
\newcommand{\tpsi}{{\tilde \psi}}

\newcommand{\tr}{{\mbox{\rm tr\,}}}

\newcounter{shownewstuffflag}

\newcommand{\startnewstuff}{\ifnum\value{shownewstuffflag}>0\color{blue}\fi}
\newcommand{\finishnewstuff}{\ifnum\value{shownewstuffflag}>0\color{black}\fi}
\newcounter{oldeq}

\newcounter{mnotecount}[section]
\renewcommand{\themnotecount}{\thesection.\arabic{mnotecount}}
\newcommand{\mnote}[1]
{\protect{\stepcounter{mnotecount}}$^{\mbox{\footnotesize $%\!\!\!\!\!\!\,
\bullet$\themnotecount}}$ \marginpar{%\color{red}%
\raggedright\tiny\em
$\!\!\!\!\!\!\,\bullet$\themnotecount: #1} }

\def\beq{\begin{equation}}
\def\eeq{\;. \end{equation}}

\newcommand{\zh}{\,{\mathring{\! h}}}
\newcommand{\zg}{{\mathring g}}
\newcommand{\eq}[1]{(\ref{#1})}

\begin{document}
\title{Initial data sets with  ends of cylindrical type:\\ I. The Lichnerowicz equation}
\author{Piotr T. Chru\'sciel\thanks{Supported in part
by the Polish Ministry of Science and Higher Education grant Nr
N N201 372736 and by IHES, Bures-sur-Yvette.}
\\  IHES, Bures-sur-Yvette, and University of Vienna
\\
\\
Rafe Mazzeo \thanks{Supported in part by the NSF grant DMS-1105050}
\\ Stanford University
}

\maketitle\thispagestyle{fancy} \rhead{\bfseries
UWThPh-2011-43}

\abstract{We construct large classes of vacuum general
relativistic initial data sets, possibly with a cosmological
constant $\Lambda\in \R$, containing ends of cylindrical type.}

\tableofcontents

% \ptc{see section 4 of Rafe uhlenbeck daniel for the
% liberalization along Delaunays; notons que la proposition 11.3.1
% des notes de Pacard donne le theoreme d'isomorphisme necessaire
% pour pouvoir ajouter une section qui donne des donnees
% initiales relativitistes par un argument de fonctions
% implicites}

% \ptc{add something about Bowen York data for general metrics,
% and how those lead to cylindrical ends, and add this to the
% abstract if done; I should read Hebey Pacard for the wrong sign
% case?}

\section{Introduction} \label{Sintro}
There are several classes of general relativistic initial data sets which have
been extensively studied, including:
\begin{enumerate}
\item compact manifolds,
\item manifolds with asymptotically flat ends,   and
\item manifolds with asymptotically hyperbolic ends.
\end{enumerate}
There is another type of interesting asymptotic geometry which appears naturally in
general relativistic studies, namely:
\begin{enumerate}
\item[4.] manifolds with ends of cylindrical type.
\end{enumerate}
The problem is then to construct solutions of the general relativistic vacuum constraint equations:
\begin{equation}
\begin{split}
& R(g) = 2\Lambda+ |K|_g^2-(\tr _g K)^2 \\
& \mbox{div}_g \, K + \nabla \tr_g K = 0  \;,
% &   D_j K^j{}_k - D_k K^j{}_j = 0  \;,  &
\end{split}
\label{8XI0.2}
\end{equation}
%\eeal{8XI0.3}
%
so that the initial data set $(M,g,K)$
contains ends of cylindrical type.

The simplest such solution is the cylinder $\R \times S^{n-1}$ with the standard product metric
and with extrinsic curvature tensor $K \equiv 0$. Its vacuum development when the cosmological constant $\Lambda$ is
zero is the interior Schwarzschild solution, and when $\Lambda > 0$ the Nariai solution. When $\Lambda = 0$, other
examples are provided by the static slices of extreme Kerr solutions or of the Majumdar-Papapetrou solutions
(see Appendix~\ref{A5I1.2}), as well as those CMC slices in the Schwarzchild-Kruskal-Szekeres space-times
which are asymptotic to slices of constant area radius $r<2m$.

Data of this type have already been studied in~\cite{GabachDain1,GabachDain2,HannamHusaNiall,GabachClement,Waxenegger,BMW,HannamEtAlPRL,BaumgarteNaculich},
but no  systematic analysis exists in the literature. The object of this paper is to initiate such a study.
As is already known from the study of the Yamabe problem on manifolds with ends of cylindrical type,
it is natural and indeed necessary to consider not only initial data sets where the metric is asymptotically cylindrical,
but also metrics which are asymptotically periodic. For the constant scalar curvature problem, these are
asymptotic to the \emph{Delaunay} metrics, cf.~\cite{Byde,CPP,ChPollack,KMPS,MPU1}), which in the relativistic
setting are the metrics induced on the static slices of the maximally extended Schwarzschild-de Sitter solutions.
We refer also to~\cite{SchoenSingular,P93,MPU1,MPU2,MPa,MPo,KMPS,Rat,Byde,CPP,ChPollack,CaffarelliGidasSpruck,ChenLin95,Marques}
for the construction, and properties, of complete constant positive scalar curvature metrics with \emph{asymptotically Delaunay} ends.
Exactly periodic solutions of \eqref{8XI0.2} are obtained by lifting solutions of the constraint equations
from $S^1 \times N$ to the cyclic cover $\RR \times N$, where $N$ is any compact manifold. In particular,  the
lifts to $\R \times S^2$ of initial data sets for the Gowdy metrics on $S^1\times S^2$ provide a large family of non-CMC
periodic solutions.

For all of these reasons, we include in our general considerations
\begin{enumerate}
\item[5.] manifolds with \emph{asymptotically periodic} ends.
\end{enumerate}
%
% \add{14II}
More generally still, some of the analysis here applies to the class of
\begin{enumerate}
\item[6.] manifolds with \emph{cylindrically bounded} ends.
\end{enumerate}
By this we mean that on each end, the metric is uniformly equivalent to a cylindrical metric,
with uniform estimates on derivatives up to some order. This is a particularly useful category because
it includes not only metrics which are asymptotically cylindrical or periodic, but also metrics which
are conformal to either of these types, with a conformal factor which is bounded above and below.
In fact, the solutions to \eqref{8XI0.2} we obtain starting from background asymptotically cylindrical
or periodic metrics are usually of this  {\it conformally asymptotically cylindrical or periodic} type,
so it is quite natural to include this broader class into our considerations to the extent possible.
We shall often refer to the geometries in cases 4)-6) as being~\emph{ends of cylindrical type}.

Finally, we shall incorporate this study of solutions with ends of cylindrical type into the more familiar
study of solutions of \eqref{8XI0.2} with ends which are asymptotically Euclidean or asymptotically
hyperbolic. In other words, we are interested in finding solutions with ends of different types,
some cylindrical and others asymptotically Euclidean or asymptotically hyperbolic.
There is a slight generalization of asymptotically Euclidean geometry which is easy to include, namely
\begin{enumerate}
\item[7.] manifolds with \emph{asymptotically conical} ends.
\end{enumerate}
This means that the metric on that end is asymptotic to $dr^2 +r^2 h$ as $r \to \infty$, where
$(N,h)$ is a compact Riemannian manifold.  The standard asymptotically Euclidean case occurs
when $N = S^{n-1}$ and $h$ is the standard round metric.  For simplicity, we use either of these
designations, i.e.\ either asymptotically Euclidean or conic, with the understanding that the results
apply to both unless explicitly stated otherwise.

The reader should keep in mind that there is experimental evidence suggesting that we live in a world with strictly
positive cosmological constant $\Lambda$~\cite{WoodVasey:2007jb,Riess:2006fw,Komatsu:2010fb},
whence the need for a better understanding of the space of solutions of Einstein equations with $\Lambda>0$.
In the simplest time-symmetric setting, this leads to the study of manifolds with constant positive scalar curvature
in which case, as discussed above, manifolds with ends of cylindrical type appear as natural models.
We note that the topology of manifolds carrying complete metrics with positive scalar curvature,
bounded sectional curvature and injectivity radius
bounded away from zero has been recently classified in~\cite{BeBeMaillot} using Ricci flow. Any such manifold
which is also orientable is a connected sum of copies of $S^1 \times S^2$ and quotients of $S^3$ by finite rotation groups;
if the manifold is noncompact, then infinite connected sums may occur.

We shall be using the standard conformal method. This is well understood within the class of CMC initial data on compact
manifolds~\cite{Jimconstraints}, at least when the trace $\tau = g^{ij}K_{ij}$ is large enough, in the sense that
\bel{8XI0.1}
\tau^2 \ge  \frac  {2n}{(n-1)}\Lambda  \;. %  \  \mbox{where}\ \tau:= g^{ij}K_{ij}  \;.
\ee
(There are now a number of results in various geometric settings which relax this condition,
see \cite{HPP,HNT3,MaxwellNonCMC,DGH,GicquaudS,VinceJim:noncmc} and references therein, but this more general problem is still far from settled.)   Since this paper is meant
to be a preliminary general investigation of the constraint equations for cylindrical geometries, we shall be assuming
the condition \eqref{8XI0.1} in almost all of the results below. Note that in the special case where $K=\frac{\tau}{n} g$,
i.e.\ the extrinsic curvature is pure trace,  insertion of \eqref{8XI0.1} into \eq{8XI0.2} yields $R(g)\leq 0$. We are \emph{not},
however, assuming that $K$ is pure-trace \emph{nor} that $R(g)\leq 0$ unless explicitly indicated otherwise in some cases below.

We recall that for the constant scalar curvature equation, the barrier method, i.e.\ the use of sub- and supersolutions, is effective
when $R(g) \leq 0$; but that when studying noncompact metrics with constant positive scalar curvature it seems to be necessary
to use more intricate parametrix-based methods.  Because we do assume \eqref{8XI0.1} throughout, we are able to use barrier methods
in all that follows.  We hope to return in future work to the study of \eqref{8XI0.2} for manifolds with ends of cylindrical
type where \eqref{8XI0.1} does not hold and more complicated analytic methods become necessary.

Our focus here is specifically on the Lichnerowicz equation and not the vector constraint equation. Thus in all that follows we
simply assume the existence of $TT$ tensors satisfying various asymptotic conditions. Large classes of such tensors are
constructed in an accompanying paper~\cite{CMP}. Alternatively, we could also use in some of our results the compactly
supported $TT$ tensors constructed in~\cite{ErwannTT}.
Many of the results below are phrased using a somewhat general form of the semilinear elliptic scalar constraint equation,
and depend only on certain structural features of the Lichnerowicz equation. In particular, we carry out our investigations {\it without}
assuming that $d\tau=0$, in hopes
that our results can, for example, be adapted for use as in \cite{DGH,HPP}. The results here directly provide solutions of the constraint
equations only when $\tau $ is constant.

As a guide to the rest of the paper, the next section reviews the Lichnerowicz equation and its behaviour under
conformal rescaling, as well as the monotone iteration scheme and the method of sub- and supersolutions.
Section 3 describes various solutions of the Lichnerowicz-type equation which depend on only one variable.
There are several cases, depending on the signs, which we assume to be constant, of the different coefficient functions.
These are the basic models for the geometry of solutions on ends of cylindrical type.

We call special attention to the class of periodic
solutions which generalize the constant scalar curvature Delaunay metrics. These {\it constraint Delaunay
solutions}, which arise when the TT tensor has constant nonzero norm, should be regarded as deformations
of the standard Delaunay solutions. Continuing on, the main
analysis is contained in \S 4.
 We separate this into two cases, depending on whether the scalar curvature
of the initial metric is nonnegative,   or else strictly negative  outside a compact set. We describe the relationships
of these conditions with the sign of the Yamabe invariant of the noncompact manifold, and prove a number of existence
theorems for solutions of \eqref{8XI0.2} with ends of cylindrical type in these two cases. The following two sections indicate how
to obtain solutions which have some cylindrical ends and others which are either asymptotically hyperbolic (\S 5)
or else asymptotically conical (\S 6).

In all cases, the existence proofs rely on construction of suitable barrier functions. Appendix A contains a proof of
existence of solutions of the Lichnerowicz-type equation using the monotone iteration scheme and assuming the existence of
suitable barrier functions, without any completeness of uniformity assumptions. We also describe there some generalities
about barrier functions and
give examples for the special geometries of interest here. Finally, Appendix B gives a number of examples from
the relativity literature where initial data sets with ends of cylindrical type are encountered.

\section{The Lichnerowicz equation}   \label{sSLe}

Fix a cosmological constant $\Lambda$ and a symmetric tensor field $\tL$ on a Riemannian manifold
$(M,\tg)$, as well as a smooth bounded function $\tau$. This last function represents the trace of the
extrinsic curvature tensor. % and is constant in the constant-mean curvature (CMC) case, but
As noted earlier, we do not assume that $d\tau = 0$. % unless explicitly indicated otherwise.
We also do not need to assume that $\tL$ is transverse-traceless, though it is in the application
to the constraint equations. To simplify notation, set
\bel{conf216} \tsigma ^2:=\frac
 {n-2}{4(n-1)}|\tL|_\tg^2\;,\quad  \beta:=\left[\frac{n-2}{4n}\tau^2 - \frac
 {n-2}{2(n-1)}\Lambda\right],\ee
which is convenient since only these quantities, rather than $\tL$, $\tau$ or $\Lambda$, appear in the constraint
equations. The symbol $\tsigma^2$ is meant to remind the reader that this function is nonnegative, but is slightly
misleading since  $\tsigma ^2$ may not be the square of a smooth function. All of these functions are as regular as the metric
$g$ and the extrinsic curvature $K$.

We shall be studying the \emph{Lichnerowicz equation}
\bel{conf214a} \Delta_\tg \tphi - \frac {n-2}{4(n-1)}\tR \tphi =
\beta\tphi^{(n+2)/(n-2)} - \tsigma ^2 \tphi^{(2-3n)/(n-2)},
\ee
which corresponds to the first of the two equations in \eqref{8XI0.2}. We write it more simply as
\bel{conf214}
L_{\tg} \tphi = \beta \tphi^\alpha - \tsigma^2 \tphi^{-\gamma}.
\ee
Here and for the rest of the paper, $L_\tg$ denotes the conformal Laplacian,
\[
L_{\tg} = \Delta_{\tg} - \frac{n-2}{4(n-1)} \tR,
\]
and we also always set
\begin{equation}
c(n)  = \frac{n-2}{4(n-1)}, \quad \alpha = \frac{n+2}{n-2} \quad \mbox{and}\quad \gamma = \frac{3n-2}{n-2}.
\label{constants}
\end{equation}

An important property of \eqref{conf214} is the following conformal transformation property. Suppose
that $\hat{g} = u^{\frac{4}{n-2}}\tg$. It is well known that for any function $\hat{\phi}$ it holds that
\[
L_{\tg}( u \hat{\phi}) = u^{\alpha} L_{\hat{g}} \hat{\phi},
\]
where $L_{\hat{g}}$ is the conformal Laplacian associated to $\hat{g}$. Thus if we set $\tphi = u \hat{\phi}$ into
\eqref{conf214},   this last equation becomes
\[
u^{\alpha} L_{\hat{g}} \hat{\phi} = \beta u^\alpha \hat{\phi}^\alpha - \tsigma^2 u^{-\gamma} \hat\phi^{-\gamma}.
\]
Dividing by $u^\alpha$ and defining
\begin{equation}
\hat{\sigma}^2 = u^{-\gamma - \alpha} \tsigma^2,
\label{hatsig}
\end{equation}
then we have simply that
\begin{equation}
L_{\hat{g}} \hat{\phi} = \beta \hat{\phi}^\alpha - \hat{\sigma}^2 \hat{\phi}^{-\gamma}.
\label{transfLich}
\end{equation}
Note in particular that while $\tsigma^2$ transforms by a power of $u$, the coefficient function $\beta$
is the same in the transformed equation.

We shall be seeking solutions $\tphi$ to \eqref{conf214} with controlled asymptotic behavior
in the following cases:
\begin{enumerate}
 \item[A.] $(M,\tg)$ is a complete manifold with a finite number of ends of cylindrical type;
 \item[B.] $(M,\tg)$ is a complete manifold with a finite number of asymptotically conic ends
and a finite number of ends of cylindrical type, with $\tR\ge 0$;
\item[C.] $(M,\tg)$ is a complete manifold with a finite number of  ends of cylindrical type
and a finite number of asymptotically hyperbolic ends. In this case we further assume that $\tR < 0$
sufficiently far out on all ends, and hence require that
$\beta > 0$ if $\tau = \const$.
\end{enumerate}

As explained in the introduction, we use the method of sub- and supersolutions throughout.
Recall that $\tphi_+$ is a supersolution of \eq{conf214} if
\bel{conf214sup}
L_{\tg} \tphi_+ \le  -\tsigma ^2 \tphi_+^{-\gamma} + \beta\tphi^{\alpha}_+\;,
\ee
while $\tphi_-$ is a subsolution of \eq{conf214} if
 \bel{conf214sub}
L_{\tg} \tphi_- \ge -\tsigma ^2 \tphi_-^{-\gamma} + \beta\tphi^{\alpha}_-\;.
\ee
We review in Appendix A below the proof of the monotone iteration scheme, which
uses a sub- and supersolution to produce an actual solution, without assuming any asymptotic conditions
on the metric.  This only requires knowledge
of the solvability properties of linear equations of the type $\Delta_{\tg} - h$,
where $h \geq 0$, on compact manifolds with boundary.

There are quite a few separate cases of \eqref{conf214} to consider, depending on whether $\tR < 0$
or $\tR \geq 0$ on the ends, on whether $\tsigma^2 \equiv 0$ or not, on the sign of $\beta$, and finally on
the different types of asymptotic geometries described above.  Some of these \emph{cannot} occur
simultaneously under our hypotheses. For example, asymptotically conic or
asymptotically hyperbolic ends, and different choices of $\beta$, preclude certain of these from occurring.
Since we are not attempting to be encyclopedic here, we shall focus on some of the key combinations
of hypotheses and omit discussion of others.  In particular, we shall always assume that
\bel{8XI0.5}
 \tsigma ^2 \not \equiv 0.
\ee
The case $\tsigma^2 \equiv 0$ is the Yamabe problem.  In the case of ends of cylindrical type,
this requires significantly different techniques than the ones used here. Note finally that
we do not discuss in any detailed way cases where $\tR$ changes sign on the ends; this occurs
in many important examples, but requires much more work to understand properly.

\section{Radial solutions}
 \label{sS10XI.1}
In this section we study the equation \eqref{conf214} on the product cylinder $(\R \times N, \tg =
dx^2 + \mathring g)$, where $\mathring g$ is a metric with constant scalar curvature $\mathring R$
on the compact manifold $N^{n-1}$. We assume further that $\tsigma ^2$ is a (positive)
constant, and for the duration of this section we also assume that $\beta$ is a constant, not necessarily positive.
Note that all of these assumptions occur naturally for the constraint equations; indeed,
$$
\tilde L =  \lambda (dx^2 - \frac 1 n \tilde g)
$$
is transverse traceless for any $\lambda\in \R$, and $|\tilde L|^2_\tg$ is
constant. %hence the condition that $\tsigma^2$ is constant is a viable one for the actual constraint equations.

Our goal here is to exhibit the wide variety of global radial solutions of this problem on
the cylinder, where by radial we mean solutions which depend only on $x$.  These are nothing
more than solutions of the ODE reduction of \eqref{conf214}.  Their importance is that each of these
solutions can occur in the asymptotic limiting behavior of a more general solution on a manifold
with ends of cylindrical type. As we have already noted, there are many types of solutions listed here,
and we shall focus on just a few of these in this paper.

Before embarking on this analysis, let us note that when $(N,\zg)$ is $S^{2}$ with its standard round metric,
we can invoke the Birkhoff theorem, or its generalizations which incorporate the case $\Lambda \neq 0$
  (see, e.g., \cite{Stanciulescu}), to conclude that the associated space-time
evolutions belong to the Schwarzschild-Kottler family; these are also known individually as the
Schwarzschild-Tangherlini, or Schwarzschild-de Sitter, or Schwarzschild-anti de Sitter solutions. We refer to all of these special solutions as the Kottler space-times.  Hence in this case, the initial data sets
described in this section yield an exhaustive description of the metrics on spherically symmetric CMC hypersurfaces
in the Kottler space-times~\cite{Estabrook,NiallEdwardCMCSchwarzschild}. In particular all the solutions for which
the conformal factor tends to zero at some finite value of $x$ lead to space-times which become singular
except if they correspond to Minkowski, de Sitter or anti-de Sitter space-time.

For functions depending only on $x$, the Lichnerowicz equation reduces to
\bel{8XI0.9}
\tphi'' -a \tphi  = - \tsigma ^2 \tphi^{-\gamma} + \beta\tphi^{\alpha},
\ee
where $a = c(n)  \tR$ is, by assumption, constant. Solutions correspond to the motion of a particle in the potential
$$
V(\tphi) = -\frac a 2 \tphi^2 - \frac{\tsigma ^2}{ \gamma-1}\tphi^{1-\gamma} - \frac \beta {\alpha+1} \tphi ^{\alpha+1};
$$
the equivalent phase space formulation uses the Hamiltonian
\[
 H(\tphi,\tpsi) :=
 \frac12 \tilde\psi^2 - \frac{a}{2}\tphi^2  - \frac{\tsigma ^2}{\gamma-1}\tphi^{1-\gamma}
- \frac{\beta}{1+\alpha} \tphi^{1+\alpha}.
\]
If $\tpsi = \dot{\tphi}$ where $\tphi$ is a solution, then $H(\tphi(t), \dot{\tphi}(t))$ is independent of $t$, or in other words,
the pair $(\tphi(x), \tpsi(x))$ remains within a level set of $H$. Since our interest is in solutions which remain bounded
away from zero and infinity for all $x$, or even better, solutions which are periodic, it is convenient to study these level sets.

Since
\[
\nabla H = ( -a \tphi + \tsigma ^2 \tphi^{-\gamma} - \beta \tphi^\alpha,  \tpsi),
\]
the critical points are all of the form $(\tphi_0,0)$, where
\begin{equation}
f(\tphi_0) := \tsigma ^2 \tphi^{-\gamma-1}_0 - \beta \tphi^{\alpha-1}_0 = a.
\label{crit}
\end{equation}
We record also that the Hessian of $H$ equals
\[
\nabla^2 H =
\left( \begin{matrix}  -a - \tsigma ^2 \gamma \tphi^{-\gamma - 1}-\alpha\beta \tphi^{\alpha-1} &  0 \\
0 & 1 \end{matrix}
\right)
 \;.
\]

The following discussion is organized  into six cases: $\beta >0$, $= 0$ and $< 0$, and $\tsigma^2 = 0$ or
$\tsigma^2 > 0$; within each of these we consider $a < 0$ and $a \geq 0$,  either separately or together
($a=0$ and $a>0$ always behave qualitatively the same). This covers the twelve possible situations.
We are, of course, only interested in solution curves in the right half-plane where $\tphi \geq 0$.
When $\tsigma^2 = 0$, the origin $(0,0)$ is always a critical point of $H$, and its presence influences
the behaviour of nearby solution curves. In this case, any solution curve which reaches the
line $\tphi = 0$ away from the origin does so at some finite value of the parameter $x$ (which
we refer to as time here). When $\tsigma^2 > 0$, then $H$ is not defined on the line $\tphi = 0$,
and no solution curve reaches this line in finite time (this uses the fact that the number $\gamma$
appearing in the exponent is greater than $2$). We omit discussion of critical points and (portions of)
level curves in the open left half-plane. We illustrate many of these cases with diagrams which exhibit
typical level curves of the function $H$. The notation $H = \mbox{crit}$ indicates a level set which contains
a critical point, and similarly $H < \mbox{crit}$ or $H > \mbox{crit}$ indicates sub- or supercritical level sets.

\begin{itemize}
\item[i)] $\beta > 0$ and $\tsigma^2  = 0$: in this case, the function $f$ is monotone decreasing, with
\[
\lim_{\tphi \to 0} f = 0 \quad \mbox{and} \qquad \lim_{\tphi \to \infty} f = -\infty.
\]
Thus \eqref{crit} has no solution when $a \geq 0$, while if $a < 0$, then there is a unique positive solution $\tphi_0$,
which is always unstable. There is also a critical point at $(0,0)$, the stability of which is determined by the sign of $a$.

Representative level curves of $H$ are shown in Figure~\ref{Fig1}; the left plot illustrates the situation $a < 0$
and the right one $a \geq 0$. For any $a$ there exist curves defined for all $x$; these are either unbounded
in $\tphi$ in both directions or else are unbounded in one direction and tend to the critical point
(either $(0,0)$ or $(\tphi_0,0)$) in the other. All other solution curves reach the line $\tphi = 0$ away
from the origin at finite time, and hence remain in the right half-plane only on a finite interval
or else a semi-infinite ray.

% most solution curves either reach $\tphi = 0$ at some finite
% value of $x$, or else decrease (in $\tphi$) from infinity monotonically to a unique minimum, then increase again
% out to infinity, and the corresponding solutions are
%  any solution with negative energy, i.e.\ $H(\tphi, \dot{\tphi}) < 0$,
% runs from infinity back out to infinity, with a single turning point; any solution with non-negative energy
% connects the origin to infinity. In the right plot of Figure~\ref{Fig1}, any orbit with energy greater
% than $H(\tphi_0,0)$ connects the origin to infinity. At the critical energy level itself, there is the constant
% solution $(\tphi_0,0)$ as well as orbits which asymptotically connect the critical point to infinity, or to
% the origin. Subcritical energies may also correspond to orbits connecting the origin to itself, or else infinity
% to itself, depending on the initial value of $\tphi$.  \rma{this is all very tied to the $V$ picture rather than
% the $H$ picture, hence confusing!}

\begin{figure}
\includegraphics[width=.5\textwidth]{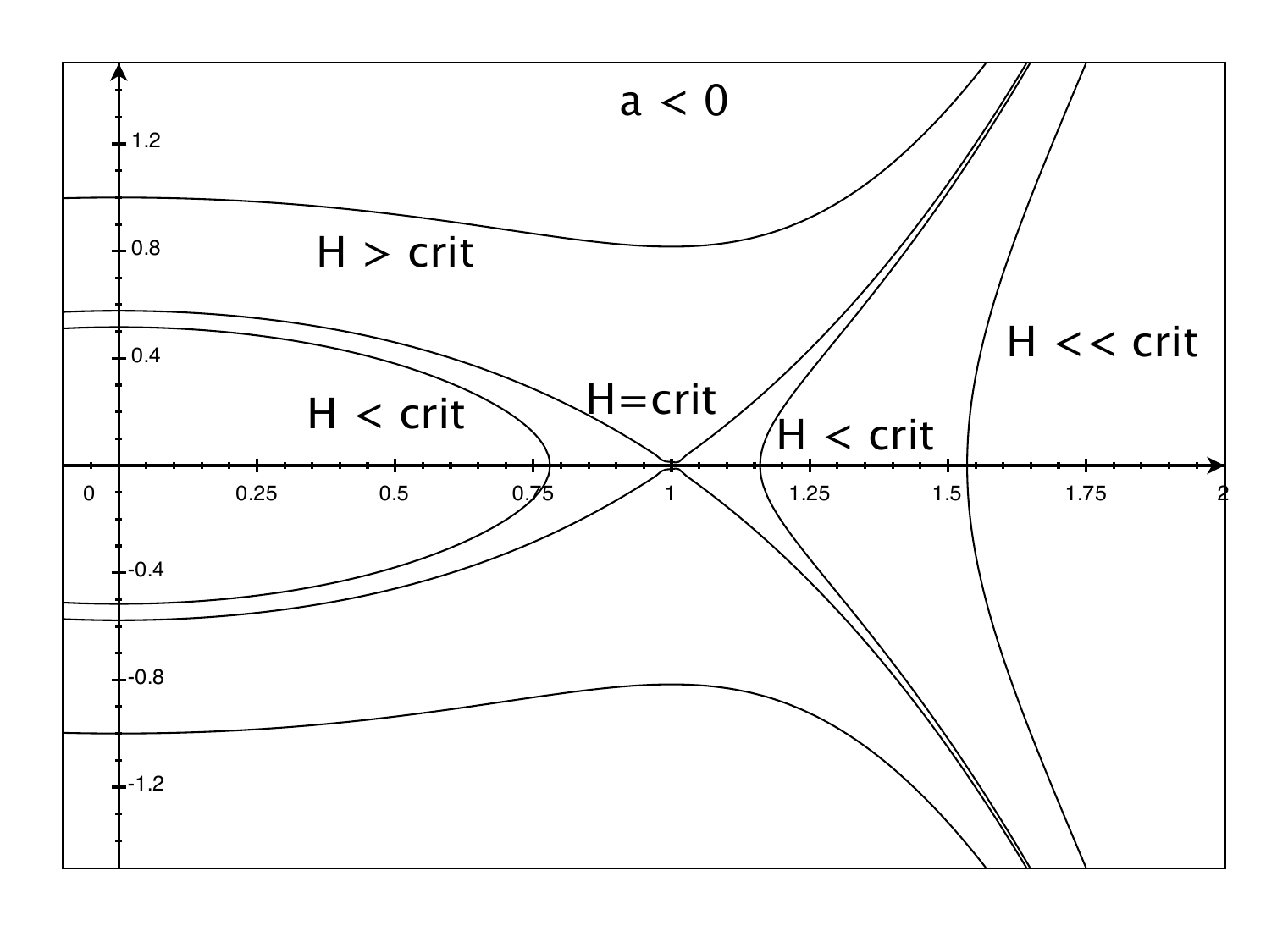}
\includegraphics[width=.5\textwidth]{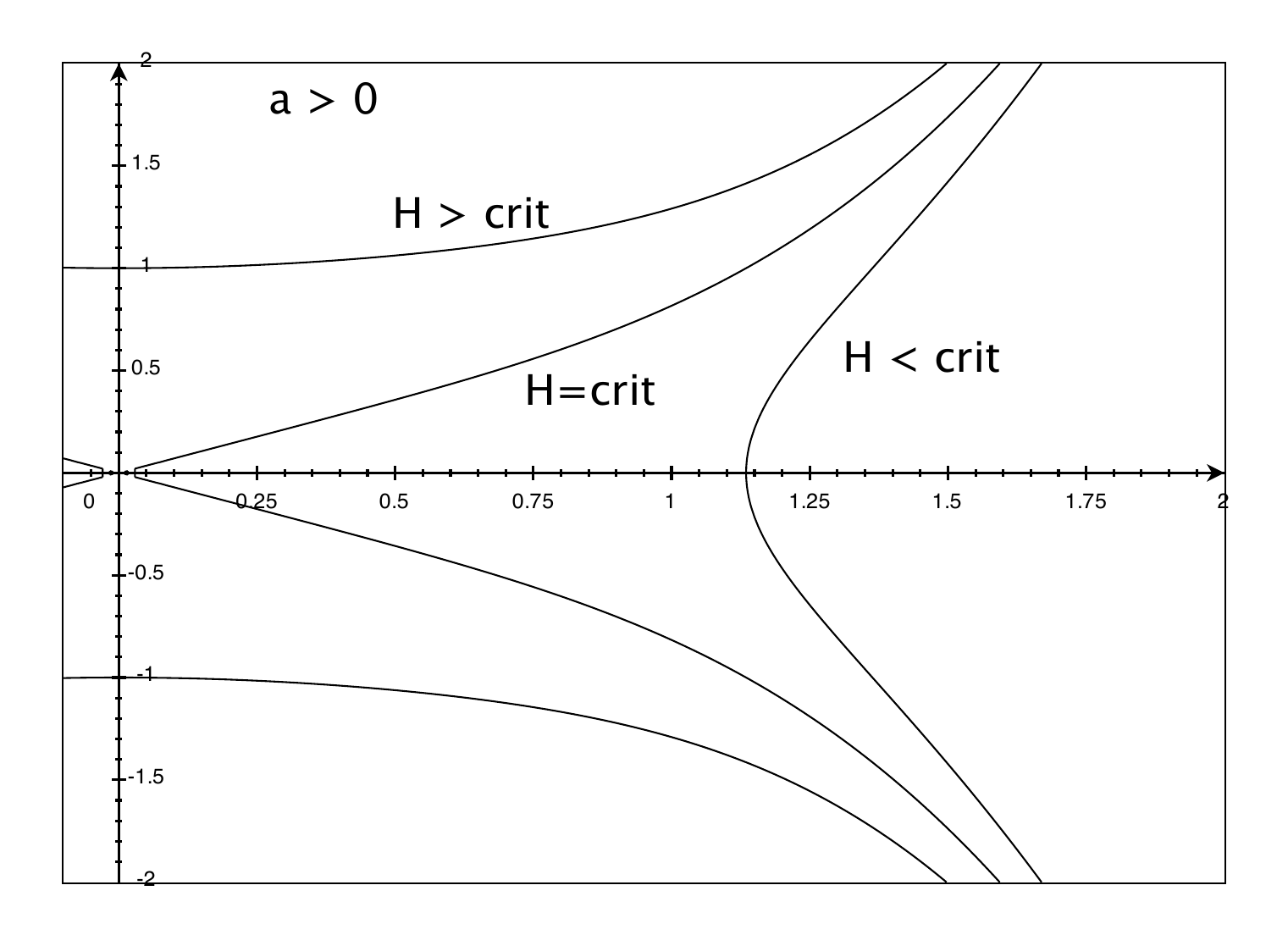}
\caption{\label{Fig1}Typical level curves of $H$ when $\beta > 0$, $\tsigma^2 = 0$}
\end{figure}

% \begin{figure}[th]
% \begin{center}
% \includegraphics[width=.4\textwidth]{potentialPositiveBetaNoSigma}
% \includegraphics[width=.4\textwidth]{potentialPositiveBetaNoSigmaNegativea}
%  % Graph: 0x0 pixel, 300dpi, 0.00x0.00 cm, bb=0 0 290 287
% \end{center}
% \caption{\label{F8XI0.2x}Representative forms of the potential $V(\tphi)$ when $\tsigma ^2= 0$ and  $\beta> 0$
% with a) $a\ge0$   and b) $a<0$.}
% \end{figure}
%
\item[ii)]  $\beta > 0$ and $\tsigma  \neq 0$: now $f$ decreases monotonically from $+\infty$ to $-\infty$,
so there is a unique positive solution $\tphi_0$ to \eqref{crit} for any value of $a$; this is always an unstable critical point.
Referring to the plot on the left in Figure~\ref{Fig2}, in the subcritical energy levels $H < \mbox{crit}$,
$\tphi$ either tends to infinity in both directions or tends to zero in both directions. In the supercritical
energy levels, $\tphi$ tends to $0$ in one direction and to infinity in the other. The nonconstant solutions
at the critical level are asymptotic to the critical point $(\tphi_0,0)$ at one end, while $\tphi$ tends to
either $0$ or infinity at the other.

% There is a single unstable equilibrium point where $V'$ vanishes. All orbits with energy larger than the one
% of this equilibrium connect $\tphi=0$ with $\tphi =\infty$ monotonically. The critical energy level contains, besides
% the unstable equilibrium, another orbit which connects zero to the critical point, and another connecting infinity
% to the critical point. For subcritical energies, there are two kinds of orbits, ones connecting zero to zero, where $x$ ranges
% only over a finite interval, and others connecting infinity to infinity.

\item[iii)] $\beta = 0$ and $\tsigma^2  = 0$; excluding the trivial case $a=0$, we see that when $a < 0$ solution
curves are all halves of ellipses which reach the line $\tphi = 0$ in finite time (this is the heavier curve
in the right plot of Figure~\ref{Fig2}). When $a > 0$, solution curves are hyperbolae for which $\tphi $
is unbounded in either both directions, or else in only one direction and reach $\tphi = 0$ in finite
time in the other; the exceptions are the curves at critical energy which exist for all time and lie along rays,
tending toward the origin in one direction and to infinity in the other.

\begin{figure}
\includegraphics[width=.5\textwidth]{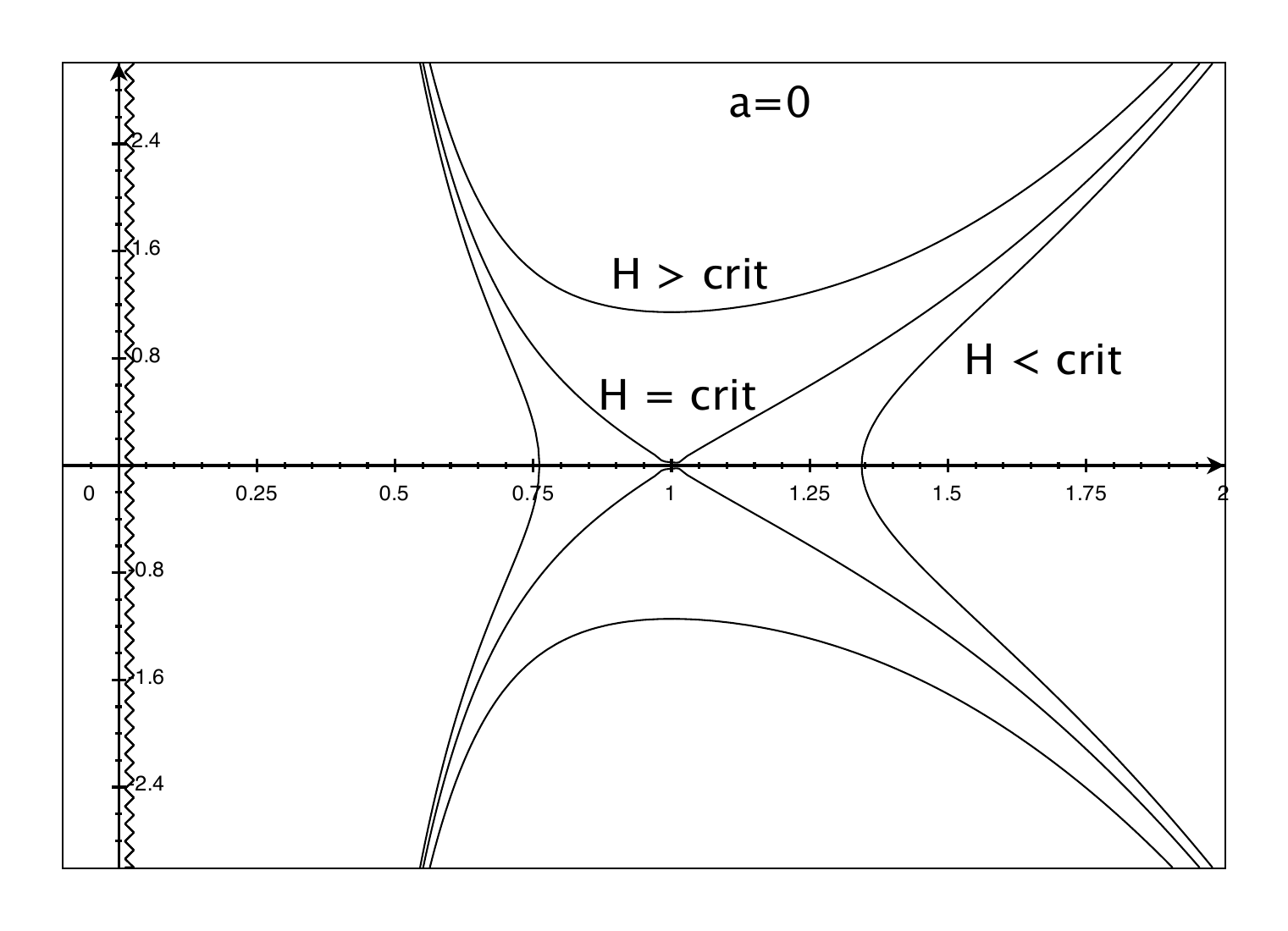}
\includegraphics[width=.5\textwidth]{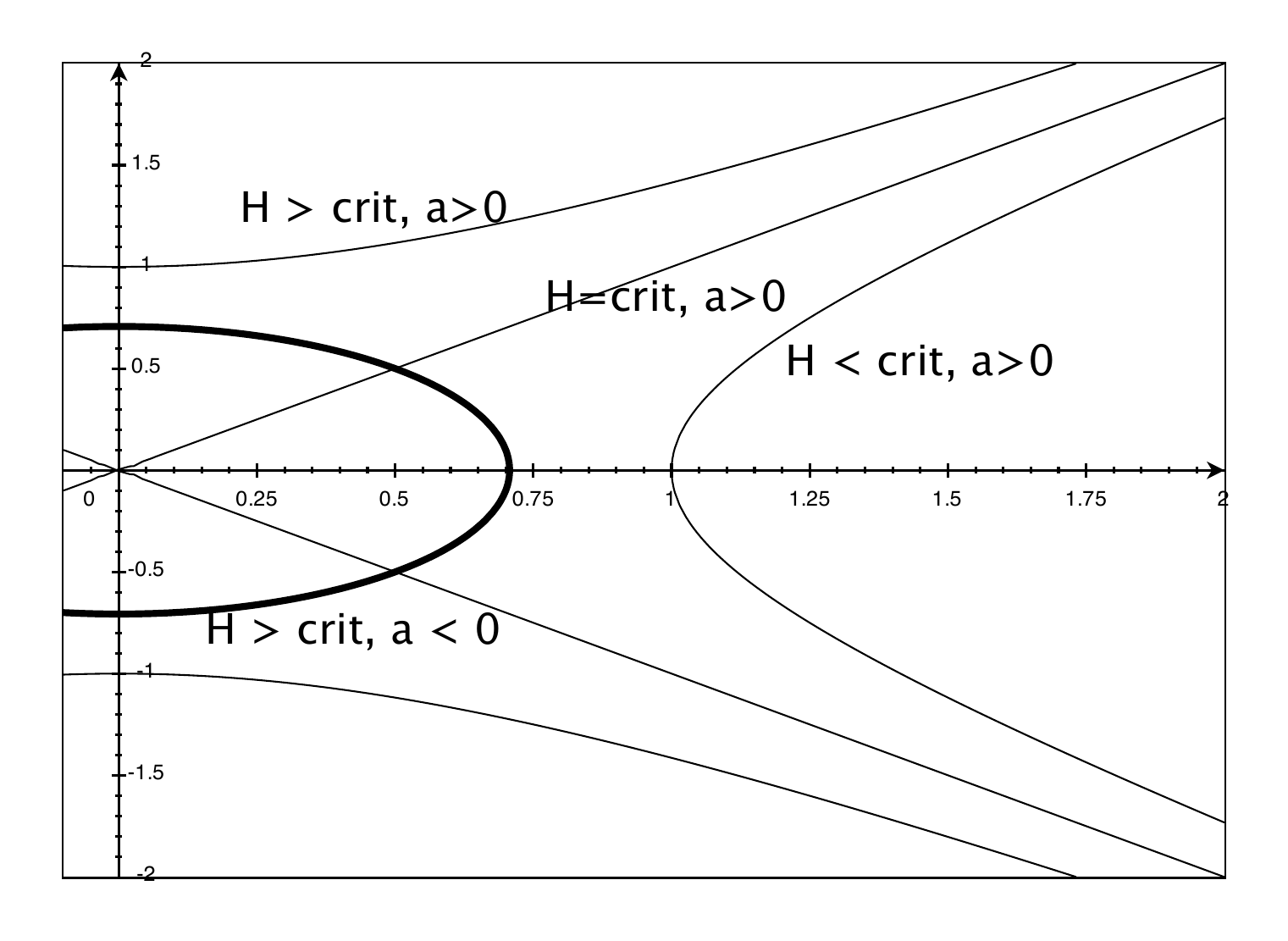}
\caption{\label{Fig2}Typical level curves of $H$ when $\beta, \tsigma ^2 >0$ ($a=0$), and when $\beta = \tsigma^2 = 0$
for $a \neq 0$; $a > 0$ corresponds to the heavier curve}
\end{figure}

\item[iv)] $\beta = 0$ and $\tsigma^2  > 0$: because $\tsigma^2$ is nonzero, the solution curves remain
in the open right halfplane for all time. Since the function $f$ decreases monotonically from $+\infty$ to $0$,
there is a critical point only when $a > 0$, and this is unstable, see the left plot of Figure~\ref{Fig3}.
All solution curves have $\tphi$ either unbounded or tending to zero in one or both directions.
When $a\ge 0$,  there is no critical point and, as in the right on Figure~\ref{Fig3}, all solution
curves have $\tphi \to 0$ both as $t \to \pm \infty$.

\begin{figure}
\includegraphics[width=.5\textwidth]{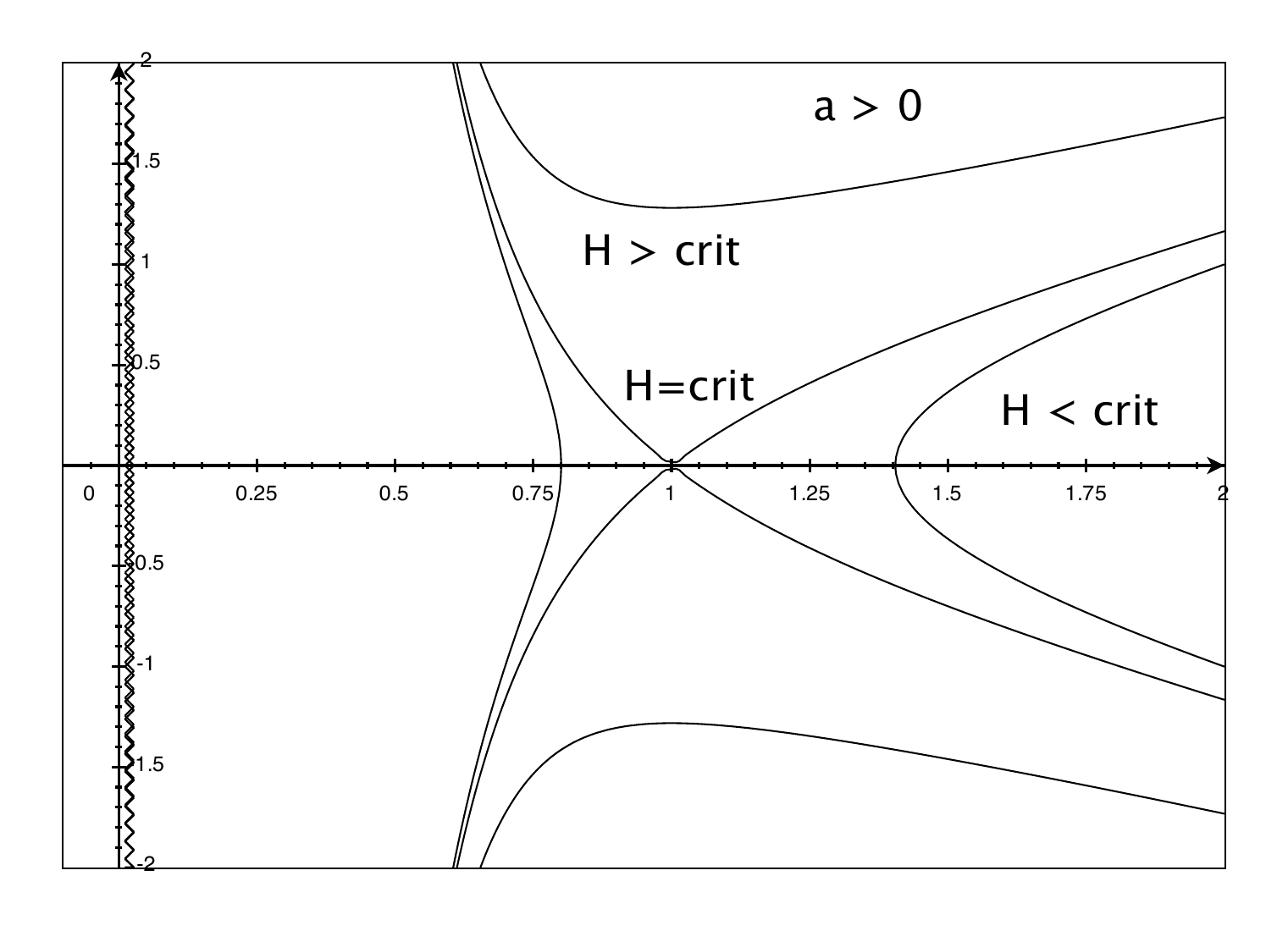}
\includegraphics[width=.5\textwidth]{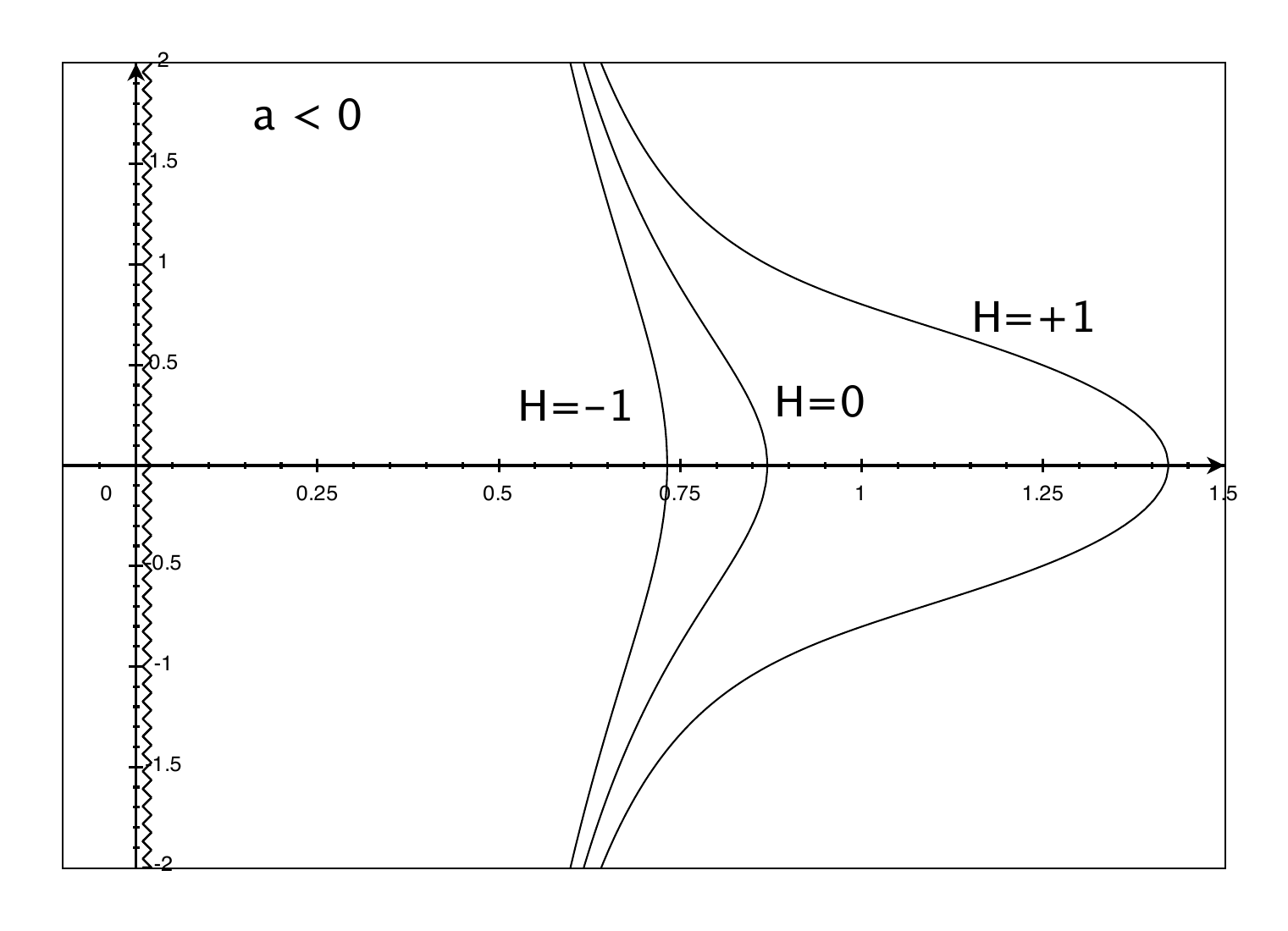}
\caption{\label{Fig3}Typical level curves of $H$ when $\beta = 0$, $\tsigma ^2 >0$, $a\neq 0$}
\end{figure}

\item[v)] $\beta < 0$ and $\tsigma^2  = 0$: there is a critical point at $(0,0)$, and another at
some point $(\tphi_0,0)$ when $a > 0$, see the left plot of Figure~\ref{Fig4}. When $a>0$, there is
a homoclinic orbit connecting $(0,0)$ to itself in infinite time in both directions. The supercritical orbits, which
lie outside this, reach $\tphi = 0$ in finite time both forwards and backwards. The orbits with
energy less than zero and greater than $H(\tphi_0, 0)$ are periodic (these are the Delaunay solutions).
When $a \leq 0$, all solutions exist only on a finite time interval.

\item[vi)] $\beta < 0$ and  $\tsigma^2 > 0$:  in this case $f$ is convex, with a unique critical point $\tphi_0$:
\[
\tphi_{\mathrm{crit}} = \left( \frac{\tsigma ^2 (\gamma+1)}{-\beta(\alpha-1)}\right)^{\frac{1}{\alpha+\gamma}} =
\left( -\frac{\tsigma ^2}{\beta} (n-1)\right)^{\frac{n-2}{4n}}.
\]
A further calculation shows that
 \bel{5XII0.1}
   f(\tphi_{\mathrm{crit}}) =
 (-\beta)^{\frac{n-1}{n}} \tsigma
^{\frac{2}{n}} (n-1)^{\frac{1-n}{n}} n =: a_0.
\ee
If $a < a_0$, equation \eqref{crit} has no solutions, and only one when $a = a_0$. In these cases, all solutions
(except $(\tphi, \tilde\psi) \equiv (\tphi_0,0)$) have $\tphi \to 0$ both as $t \to \pm \infty$.
When $a > a_0$ there are two critical points, $(\phi_0^{\pm},0)$, see the right plot of Figure~\ref{Fig4},
with $0 < \tphi_0^- < \tphi_0^+$ and $H(\tphi_0^-) < H(\tphi_0^+)$.  There is a homoclinic orbit
connecting $(\tphi_0^-,0)$ to itself and defined for all $x \in \RR$. Noncritical solution curves inside
this orbit are periodic, and we call these the {\it constraint Delaunay solutions}. Solution curves outside
the closure of this homoclinic orbit have $\tphi \to 0$ either as $x \to \infty$ or $-\infty$ or both.

\begin{figure}
\includegraphics[width=.5\textwidth]{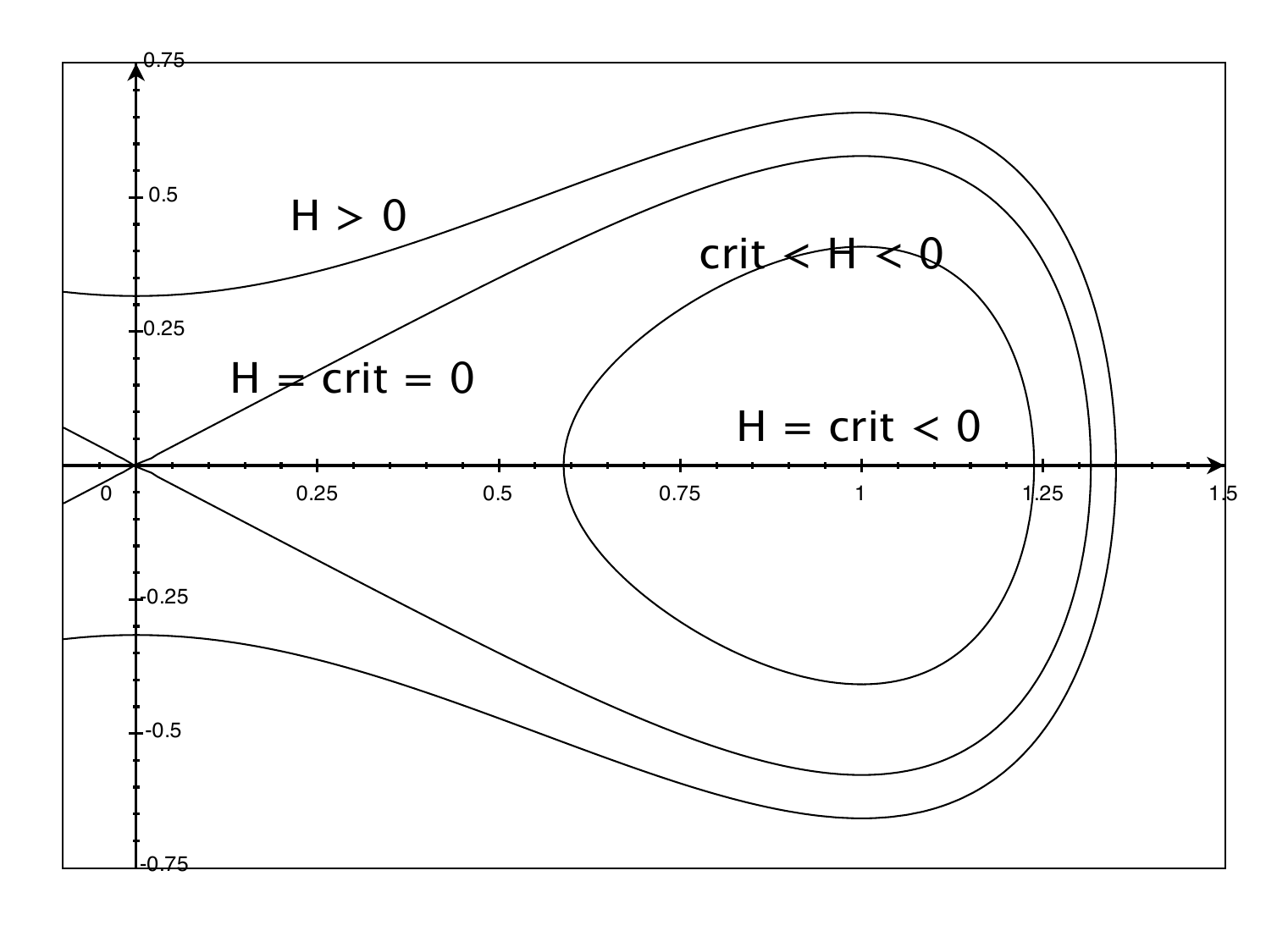}
\includegraphics[width=.5\textwidth]{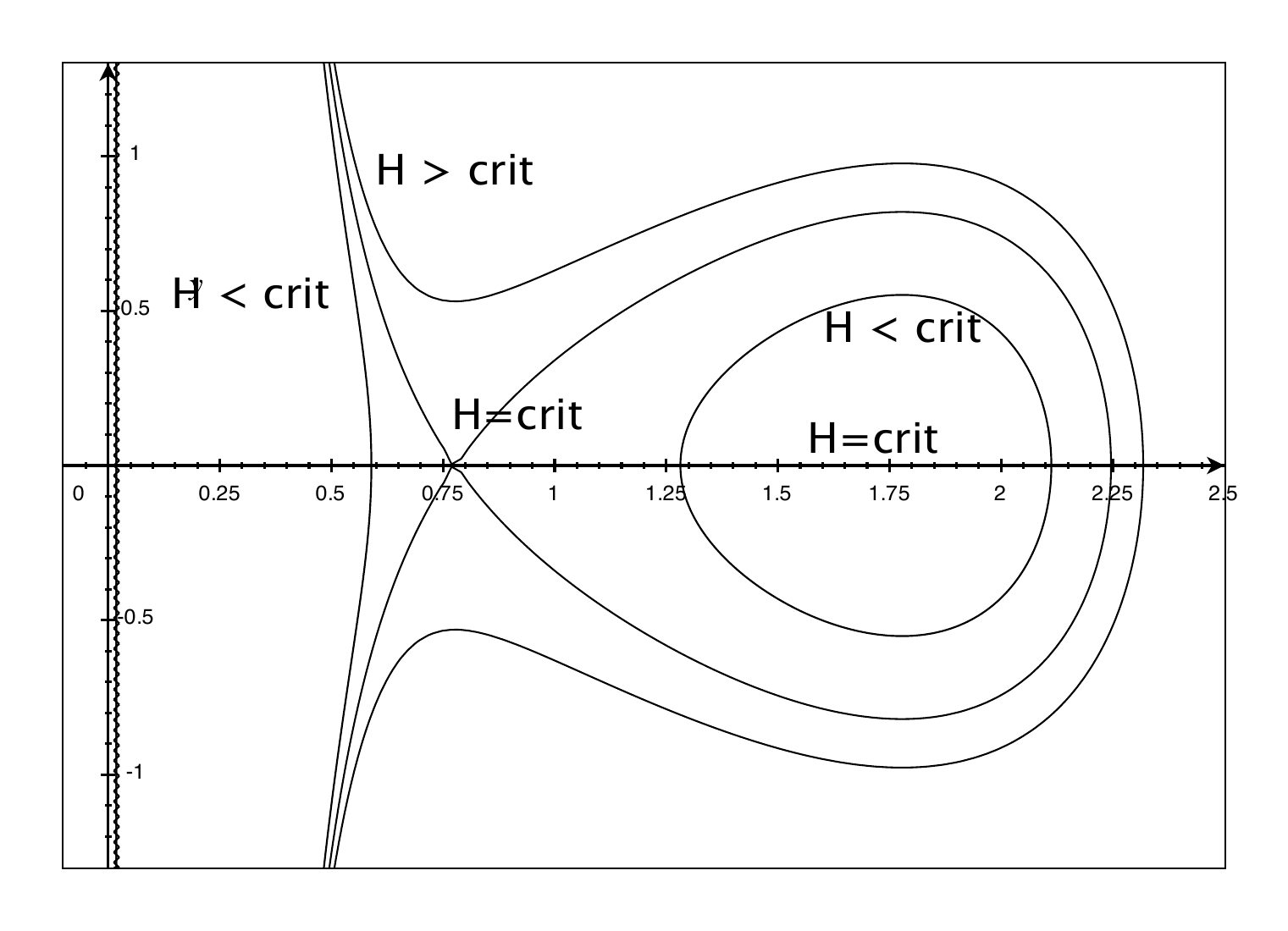}
\caption{\label{Fig4}Typical level sets of $H$ which include periodic orbits when $\beta < 0$, $\tsigma^2 = 0$ (left) and $\tsigma^2 > 0$ (right); in the central enclosed regions in each figure, H=crit denotes the stationary point $(\tphi_0,0)$.}
\end{figure}
\end{itemize}

% \begin{figure}[th]
% \begin{center}
% %\includegraphics[width=.4\textwidth]{potentialPositiveBetaNoSigma}
% \includegraphics[width=.4\textwidth]{potentialNegativeBetaNoSigma}
% \includegraphics[width=.4\textwidth]{potentialOneMore}
%  % Graph: 0x0 pixel, 300dpi, 0.00x0.00 cm, bb=0 0 290 287
% \end{center}1
% \caption{\label{F8XI0.2}Representative forms of the potential $V(\tphi)$ when $\tsigma ^2= 0$  with
%     $\beta<0$ and a) $a>0$ or b) $a<0$. The first case corresponds to Delaunay metrics.}
% \end{figure}
% %
% \begin{figure}[t]
% \begin{center}
% \includegraphics[width=.4\textwidth]{potentialSmallNegativeBetaSigmaNonZero}\end{center}
% \caption{\label{F8XI0.3}Representative form of the potential $V(\tphi)$ when $\tsigma ^2> 0$,
%  $\beta<0$ and $a>a_0$ as given by \eq{5XII0.1}.}
% \end{figure}

To summarize, the {\it only} cases where there exist solutions for which $\tphi$ remains bounded
and uniformly positive for all $x$ are v) and vi), and in these two situations, the relevant solutions
are either constant or else periodic. It is worth emphasizing that the minimum `necksize' for the family
of constraint Delaunay solutions in case vi) is strictly positive.

\section{Manifolds with ends of cylindrical type}
In this section we consider complete
manifolds $(M,\tg)$ which have a finite number of ends, each of cylindrical type.
Thus each end $E_\ell$ of $M$ is identified with the product $\RR^+ \times N^{n-1}_\ell$, where $N_\ell$ is compact.
In the most general scenario, the metric $\tg$ satisfies
\bel{18XII11.1}
c_1 (dx^2 + \zg_\ell) \leq \left. \tg \right|_{E_\ell} \leq c_2 (dx^2 + \zg_\ell)
\ee
for some positive constants $c_1, c_2$ and fixed Riemannian metric $\zg_\ell$ on $N_\ell$.  We are specifically
interested in the cases when $\tg$ on each end is conformal to a metric which is either asymptotically cylindrical or
asymptotically periodic, but the more general condition \eqref{18XII11.1} will be allowed in some of our results.
To set notation, we say that $\tg$ is asymptotically conformally cylindrical on the end $E_\ell$ if
$\tg = w^{\frac{4}{n-2}} \check{g}$ there, where
\[
\check{g}  = dx^2 + \zg_\ell + \calO(e^{-\nu x}),
\]
for some metric $\zg$ on $N_\ell$; $\tg$ is called asymptotically conformally periodic if $\check{g}$ is asymptotic at the
same exponential rate to a metric $\zg$ on $\RR \times N_\ell$ which is periodic in $x$ with period $T$.  For convenience
we often drop the index $\ell$ when the argument being made does not specifically address issues related
to the presence of more than one end.

We also suppose that the functions $\beta$ and $\tsigma^2$ converge to
functions $\zbeta$ and $\zsigma^2$ which are either independent of $x$
(in the conformally asymptotically cylindrical case) or else periodic with period $T$.
Finally, we assume in both cases that the conformal factor $w$ converges
to a function $\zw$, which is a smooth positive function on $N$.  In all cases,
convergence is at the rate $e^{-\nu x}$. While the methods here do not require much
differentiability, we do need to assume that $\tR \to \ztR$ in those results which assert existence of limits of the solution.

% {\color{blue}
%  \ptcr{this is wrong in most of our results, because such conformal transformations will not preserve our hypotheses on the initial scalar curvature. I did not do this in my proofs. If you rewrote the proofs in a way which does that, then please revert to my original proofs. This paragraph has to go in any case}
% Notice that by the conformal transformation properties of the Lichnerowicz equation described
% in \S 2, if $\tg$ is {\it conformally} asymptotically cylindrical or periodic, then we are free to make an initial
% conformal change and assume that $\tg$ is already asymptotically cylindrical or periodic.
% Sometimes a slightly stronger fact is true: if a result does not require that $\tsigma^2$ tends to an
% asymptotic limit, then it is not necessary to assume that the conformal factor $w$ tends to one.
% This is simply because in the transformed equation, $\hat\sigma^2 = w^{-\gamma-\alpha}\tsigma^2$.
% In any case, for this reason, in results below, the hypotheses include the conformal factor but the proof
% often omits it.}

We shall construct solutions $\tphi$ to the Lichnerowicz equation under various hypotheses.
There is a preliminary very general result which assumes only that $\tR$ is negative and bounded away
from $0$ far out in each end, but which requires almost nothing else about the structure of $\tg$
except that it is of cylindrical type. The rest of the results here assume some form of either of the
conditions $\tR \leq -c$ or $\tR \geq 0$.
To obtain solutions which have good asymptotic behavior, one needs to make more hypotheses than just cylindrical boundedness, and we prove such results in the asymptotically conformally
cylindrical and periodic settings. If the putative solution $\tphi$ converges to a limit $\zphi$ (which is either
a function on $N$ or else periodic on $\RR \times N$), then $\zphi$ must be a solution of the
{\it reduced Lichnerowicz equation}
\begin{equation}
\label{redlich}
\Delta_\zg \zphi - c(n) \ztR  \zphi =\zbeta\zphi^{\alpha} - \zsigma^2 \zphi^{-\gamma}.
\end{equation}
In the asymptotically periodic case, when $\zphi$ is a function on $N$, this is {\it not}
the Lichnerowicz equation on $N$, since the constants $c(n)$, $\alpha$ and $\gamma$
are for an $n$-dimensional manifold rather than an $(n-1)$-dimensional one. On the other
hand, in the asymptotically periodic case this is the standard Lichnerowicz equation on
the $n$-dimensional manifold $S^1 \times N$.
In any case, as part of the work ahead, we must prove existence of solutions to \eqref{redlich}.

In several places below we shall use the Yamabe invariant of the conformal class $[\tg]$ on $M$. This is the value
\begin{equation}
\label{25X0.1x}
\begin{aligned}
Y(M,[\tg]) =  & \inf_{u\in \calC_0^\infty\atop{\ 0 \leq u\not \equiv 0}} Q_{\tg}(u), \\
& \mbox{where} \ Q_{\tg}(u) =  \frac{ \frac12 \int_M ( |\nabla u|^2 + c(n) \tR u^2)}{\big(\int_ M u^{ \frac{2n}{n-2}} \big)^{ \frac{n-2}{n}}}.
\end{aligned}
\end{equation}
If $K$ is a compact domain in $M$ with smooth boundary, then $Y(K, [\tg])$ denotes the infimum of the same quantity
amongst smooth nonnegative functions with support in $K$.   We recall the well-known
\begin{lemma}
 \label{L19XII11.1}
Let $\lambda_0$ denote the lowest eigenvalue of $-L_{\tg}$ on $K$ with Dirichlet boundary conditions.  Then $\lambda_0 \neq 0$
if and only if $Y(K, [\tg]) \neq 0$,  and if this is the case, then the signs of $\lambda_0$ and $Y(K, [\tg])$ are the same.
\label{sgnYam}
\end{lemma}
\begin{proof}
Let $\calR_{\tg}$ denote the Rayleigh quotient for $-L_{\tg}$,
\begin{equation}
\calR_{\tg}(v) = \frac{ \frac12 \int_K ( |\nabla v|^2 + c(n) \tR v^2)}{\int_ K v^2}.
\label{Rayleigh}
\end{equation}
Since
\begin{equation}
\calR_{\tg}(v) \int_K v^2 = Q_{\tg}(v) \big(\int_ K  v^{ \frac{2n}{n-2}} \big)^{ \frac{n-2}{n} } ,
\label{RQ}
\end{equation}
% \[
% Q_{\tg}(v) = \calR_{\tg}(v) \frac {\int_K u^2}{\big(\int_ K  u^{ \frac{2n}{n-2}} \big)^{ \frac{n-2}{n} } }
% \]
we see that $Q_{\tg}(v)$ and $\calR_{\tg}(v)$ have the same sign for any $v$. {}From this it follows directly that
$\lambda_0 < 0$ if and only if $Y(K,[\tg]) < 0$.

Thus we may now assume that $\lambda_0 \geq 0$ and $Y(K, [\tg]) \geq 0$, and it remains to show that
either both of these quantities vanish, or neither of them do.

Applying H\"older's inequality to \eqref{RQ} and using that $\calR_{\tg}(v) \geq 0$ shows that
$Q_{\tg}(v) \leq C \calR_{\tg}(v)$ for some constant $C$ independent of $v$, hence if $\lambda_0 = 0$
then $Y(K, [\tg]) = 0$.

For the converse, recall the Sobolev inequality on $K$, which states that
$$
\big(\int v^{2n/(n-2)} \big)^{ (n-2)/n}  \leq A \int |\nabla v|^2 + B \int v^2\;,
$$
for some constants $A$ and $B$. Rewrite the right side as
\begin{multline*}
2A \big( \frac12 \int |\nabla v|^2 + c(n) \tR v^2)  +  \int (B-  A c(n) \tR) v^2 \\
\leq  A' \big( \frac12 \int |\nabla v|^2 + c(n) \tR v^2)  +  B' \int   v^2\;,
\end{multline*}
for some other constants $A', B' > 0$. Dividing by $\int v^2$ and using that
$Q_{\tg}(v), \calR_{\tg} (v) \geq 0$ for all $v$, we obtain
\begin{equation}
\calR_{\tg}(v) \leq Q_{\tg}(v) ( A' \calR_\tg(v) + B ') \Rightarrow \calR_{\tg}(v) \leq C' Q_{\tg}(v)
\label{Sob}
\end{equation}
for any $v \not \equiv 0$. This shows that if $Y(K, [\tg]) = 0$ then $\lambda_0 = 0$.
This proves all the assertions in the lemma. \qed
\end{proof}

In what follows we will assume that $\tR$ remains bounded away from zero outside a compact set, or has constant sign everywhere, which is certainly restrictive.
There are
many cases of interest where this fails and the limiting scalar curvature $\zR$ changes sign.  As already mentioned, those cases
are not discussed here since they require fairly different techniques.

\subsection{$\tR<0$ outside of a compact set}
 \label{ss2I1.2}
We begin with a general existence theorem assuming only that the scalar curvature
$\tR $ is negative outside of a compact set. In this section we are mainly interested in complete
manifolds with all ends of cylindrical type, but neither of these conditions is assumed in our next result:

\begin{Proposition}
 \label{P19XII0.1} Let $(M,\tg)$ contain at least one cylindrical end.
Suppose that there exist  constants $c ,C>0$ and a compact set $K$  such that
\bel{19XII11.1}
 c\le \beta\le C \;,
 \quad
-C \le \tR\big|_{M\setminus K}\le -c
% \;,
% \quad
% \frac \beta {|\tR|}\bigg|_{M\setminus K}
% \le C
 \;.
\ee
(No bound on $\tsigma^2$ is required.) Then  there exists a solution of the Lichnerowicz equation.
The conformally rescaled metric is complete if $\tg$ is, and any end  of cylindrical type remains cylindrical.
\end{Proposition}
\begin{proof}
We begin by using a result of Aviles and McOwen \cite[Thm.\ A]{aviles:mcowen:noncompact:1988} that, given
the hypotheses here, there exists a
conformal deformation $\hat{g} = u^{\frac{4}{n-2}} \tg$ such that $ \hat{R} \leq -c < 0$ everywhere.
Since the first paragraph (p.~230 in~\cite{aviles:mcowen:noncompact:1988}) of their proof is not clear to us,
we provide an alternative argument for that step of their proof.

It is not hard to see,  by the assumption on the strict negativity of $\tR$ far out on the cylindrical end in
particular, that we have $Y(K, [\tg]) < 0$.
By Lemma~\ref{sgnYam}, we must also have $\lambda_0 < 0$,
where $\lambda_0$ is the lowest eigenvalue of $-L_{\tg}$ on $K$. If $u_0$ is the corresponding eigenfunction, then
$L_{\tg} u_0 = -\lambda_0 u_0$, and this is strictly positive in the interior of $K$.  Choose a large constant $A$ so
that $A u_0 > 1$ in some slightly smaller compact set $K'$ in the interior of $K$ but $A u_0 < 1$ near $\del K$.
Then set $u = \max\{ Au_0, 1\}$ in $K$ and $u = 1$ outside of $K$.  Thus $L_{\tg} u \geq c > 0$ weakly everywhere.
Since $u$ is smooth away from
the hypersurface $H$ where $Au_0 = 1$ (and by changing $A$ slightly we can always assume that $H$ is a smooth submanifold),
it is not hard to see that this inequality is true in a distributional sense as well. Hence we can choose a mollification $\tilde{u}$
of $u$ which is smooth and satisfies $L_{\tg} \tilde{u} \geq c/2$ everywhere.
Thus if $\hat{g} = \tilde{u}^{\frac{4}{n-2}} \tg$, then
\[
\hR = -\frac{1}{c(n)} \tilde{u}^{-\alpha} L_{\tg} \tilde{u}
\]
is everywhere negative and bounded between two negative constants.%

We replace $\tg$ by the metric $\hg$ and rename it $\tg$ again; we can thus assume
that \eq{19XII11.1} holds with $K=\emptyset$.

With this reduction, we can immediately find a solution to the Lichnerowicz equation. Indeed, the bounds \eq{19XII11.1}
 imply that the constant function $\eta $ is a subsolution and $1/\eta$ is a supersolution when $0<\eta \ll 1$.  The result
now follows from Proposition \ref{P13XIII0.1} in Appendix A.
\qed
\end{proof}

If we drop the hypothesis on the lower bound $\tR\ge -C$,  then we can proceed as in the later part
of the proof of \cite[Theorem A]{aviles:mcowen:noncompact:1988} to find a conformal factor $u$ which is bounded
away from $0$ but not necessarily bounded above, so that the corresponding metric is complete, but does not necessarily
have cylindrically bounded ends. It is probable that one can go on from this to find a solution of the Lichnerowicz
equation. This certainly requires more work and we did not attempt it.

Proposition \ref{P19XII0.1} is satisfactory in its generality, but has the defect that it gives very little
information about the behavior of solutions far out in the ends. For metrics which are
asymptotically cylindrical or asymptotically periodic, we can say much more.  We take this
up in the next two subsections.

\subsubsection{Conformally asymptotically cylindrical ends}
\label{ss2I1.5}
Let $(M, \tg)$ be a complete manifold with conformally asymptotically cylindrical ends. Thus, as in the
beginning of this section, $\tg = w^{\frac{4}{n-2}}g$, and all quantities considered here
decay to asymptotic limits, which are functions or metrics on $N$, at the rate $e^{-\nu x}$
along with at least two derivatives. We assume also that in each end,
\[
\ztR\le -c \  \mbox{and} \ \beta\ge c  \ \mbox{for some $c >0$}.
\]
We now sharpen Proposition~\ref{P19XII0.1} in the sense that if we select any positive solution $\zphi$ for the
reduced Lichnerowicz equation \eqref{redlich} on each end, then there exists a solution
$\tphi$ which converges to $\zphi$ at infinity.  Hence even if we start with an asymptotically cylindrical metric,
the solution metric is conformally asymptotically cylindrical. (Note too that there are many natural example of
initial data sets with conformally asymptotically cylindrical ends, for example appropriate slices of the extremal Kerr metrics near their event horizon, see Appendix~\ref{A5I1.2}.)

The existence of a unique limit function $\zphi$ which solves \eqref{redlich} on each asymptotic cross-section
$N^{n-1}$ is straightforward: with our hypotheses on $\zR$ and $\beta$, the constants $\eta$ and $1/\eta$,
with $\eta$ sufficiently large, are sub- and supersolutions of \eqref{redlich}, and uniqueness follows from
the maximum principle.

Choose a smooth positive function $\phi$ which converges (at the rate $e^{-\nu x})$) to the limit function
$\zphi$ on each end. We shall look for solutions in the form $\phi (1 + u)$ where $u \to 0$ at infinity.
As in \S 2, we use the conformal transformation properties of the Lichnerowicz equation to write the equation
as follows. Setting $\hg = \phi^{\frac{4}{n-2}} \tg$, then
\begin{equation}
 \label{pL}
L_\hg (1+u) = \beta (1+u)^\alpha - \hsigma^2 (1+u)^{-\gamma},\ \mbox{where}\ \hsigma^2 = \tsigma^2 \phi^{-\frac{4n}{n-2}}.
\end{equation}
Taking the limit of the conformal transformation rule $L_{\tg} \phi = \phi^\alpha L_{\hg} 1$ as $x \to \infty$ gives
\[
(\Delta_{\ztg} - c(n) \ztR) \zphi = \zphi^\alpha ( -c(n) \zhR),
\]
and combining this with \eqref{redlich} for $\zphi$ results in
\[
\beta \zphi^\alpha - \ztsigma^2 \zphi^{-\gamma} = \zphi^{\alpha} (-c(n) \zhR) \Longrightarrow
c(n) \zhR  + \zbeta - \zhsigma^2 = 0,
\]
where $\zhsigma^2$ is the limit of $\hsigma^2$.  (We use here that $\del_x \phi, \del_x^2 \phi = o(1)$.)
We conclude the decay rate
\begin{equation}
 s := c(n) \hR + \beta - \hsigma^2= \calO(e^{-\nu x})
 \;.
\label{rbs}
\end{equation}

In order to construct the desired supersolution in the asymptotic region, we
rewrite \eqref{pL} using this estimate, also subtracting $\beta \alpha u$ from each side. This yields
\begin{equation}
\begin{aligned}
\big(\Delta_{\hg} - & (\hsigma^2 + \frac{4}{n-1}\beta + \calO(e^{-\nu x})) \big) u = \\
& \beta\left( (1+u)^\alpha - 1 - \alpha u\right) + \hsigma^2 \left( 1 - (1+u)^{-\gamma} \right)+ s
 \;.
\end{aligned}
\label{pL2}
\end{equation}
For simplicity, denote by $h$ the term of order zero in the linear operator on the left, i.e.
\[
h = \hsigma^2 + \frac{4}{n-2}\beta + \calO(e^{-\nu x}).
\]
Note that $h \geq c > 0$ for $x$ sufficiently large.

The motivation for this rearrangement is that the function
\[
r(u) = \beta\left( (1+u)^\alpha - 1 - \alpha u\right) + \hsigma^2 \left( 1 - (1+u)^{-\gamma} \right)
\]
satisfies $r(0) = 0$  and
is positive for $u > 0$, as can be seen directly by differentiating.

To continue, observe that given $\e > 0$ satisfying  $4\epsilon^2 < c$, we can choose $X_\e$ sufficiently large that
\[
\hsigma^2 + \frac{4}{n-1} \beta + \calO(e^{-\nu x})\geq 4\e^2\ \ \mbox{for}\ x \geq X_\e.
\]
Now, for any $0 < \delta < \e$, $\Delta_\tg e^{-\delta x} \le 2\e^2 e^{-\delta x}$, also for $x \geq X_\e$, whence
$$
\left(\Delta_{\hg}  - h\right) e^{-\delta x}  \le -2\epsilon^2  e^{-\delta x}
$$
in $[X_\e,\infty) \times N$. Finally, assuming that $\delta < \nu$, choosing
$$
\hat C \ge \frac { \| e^{\nu x} s\|_{L^\infty}}{2\epsilon^2} e^{(\delta-\nu)X_\epsilon}
$$
and $x \geq X_\e$, and using the positivity of $r$, we obtain
\[
\left(\Delta_{\hg} - h \right)  \hat C e^{-\delta x}  \le s \leq s + r(\hat C e^{-\delta x}).
\]
Comparing with \eqref{pL2}, this says simply that $1+\hat C e^{-\delta x}$ is a supersolution.

We now show that $1 - \hat C e^{-\delta x}$ is a subsolution in the asymptotic region if $\delta$ is sufficiently small.
Using \eqref{rbs}, we first calculate that
\begin{equation}
\begin{split}
(\Delta_\hg-c(n) \hR)(1-  \hat C e^{-\delta x}) =  (\Delta_\hg+ \beta-\hsigma^2 +s)(1-  \hat C e^{-\delta x}) \\
= \big( -\hat C\Delta_\hg e^{-\delta x}   +s(1-  \hat C e^{-\delta x}) \big) + \big( \beta-\hsigma^2 )(1-  \hat C e^{-\delta x})\big).
\end{split}
\label{15XII11.1}
\end{equation}
%We denote the two parenthetical terms on the right by $I$ and $II$, respectively.
Setting $y = 1 - \hat C e^{-\delta x}$, then  $y < 1$ for all $x \geq 0$ and $\hat C\ge 0$.
Fixing $\mu \in (0,1)$, $\delta \in (0,\nu)$ and any $\hat C > 0$, we can choose $X =
X(\hat C, \delta)$ so large that for all $x \geq X$, we have $y \in [\mu, 1)$, and moreover
$$
 \frac{y-y^\alpha}{1-y}\ge \mu' > 0
 \;,
$$
for some positive number $\mu'$ which depends only on $\mu$ and $\alpha$. This follows directly from the fact that the function on the left has a continuous
extension to the closed interval $[\mu,1]$ and is everywhere positive there.  This gives
\[
1 - \hat C e^{-\delta x} \geq \mu' \hat C e^{-\delta x} + (1 - \hat C e^{-\delta x})^\alpha.
\]
It is also clear that
\[
 -\hsigma^2 (1-  \hat C e^{-\delta x}) \ge  -\hsigma^2 (1-  \hat C e^{-\delta x})^{-\gamma}.
\]
Using these last two inequalities on the right side of \eqref{15XII11.1}, we see that
\[
L_{\hat g} (1 - \hat C e^{-\delta x}) \geq \beta (1 - \hat C e^{-\delta x})^\alpha - \hsigma^2 (1-\hat C e^{-\delta x})^{-\gamma} + A,
\]
where
\[
A :=  -\hat C \Delta_\hg e^{-\delta x} + s (1-  \hat C e^{-\delta x})
  + \beta \mu' \hat C e^{-\delta x}.
\]
The first term on the right here is $\calO(\delta^2 e^{-\delta x})$, while the second is bounded by $-\| e^{\nu x} s\|_{L^\infty} e^{ -\nu x}$.
Choosing $0 < \delta \ll \nu$, we certainly have that $A \geq 0$ for $x$ large enough.  This proves that
$1 - \hat C e^{-\delta x}$ is a subsolution sufficiently far out on the end.

% %
% This shows that the term $II$ in \eq{15XII11.1} satisfies the inequality needed for $1-  \hat C e^{-\delta x}$ being a subsolution:
% %
% \bel{16XII11.1}
%  I+II \equiv L_\hg (1-  \hat C e^{-\delta x})\ge  \beta (1-  \hat C e^{-\delta x} )^\alpha -\hsigma^2  (1-  \hat C e^{-\delta x})^{-\gamma}
% \;.
% \ee
% %.
%  \ptcr{added the modified construction of a subsolution }
% However, we need a stronger inequality  to account for the term $I$. For this let $0<\mu<1$ and suppose that
% %
% \bel{15XII11.2}
%  \mu <\underbrace{1-  \hat C e^{-\delta x}}_{=:y}<1
% \;.
% \ee
% %
% We claim that in this range of the variable $y$ there exists a strictly positive constant $\sigma$ such that
% %

% %
% This follows immediately from the facts that the l.h.s.\ is positive, and that the limit, at $y=1$, of the l.h.s.\ equals $\alpha-1>0$. The last inequality implies, under \eq{15XII11.2},
% %
% $$
%   \beta (1-  \hat C e^{-\delta x}) \ge   \beta (1-  \hat C e^{-\delta x} )^\alpha + \beta\sigma \hat C e^{-\delta x}
%  \;.
% $$
% %
% We see that \eq{16XII11.1} will hold
% provided that
% %
% \bel{15XII11.5}
%  I + \beta \sigma \hat C e^{-\delta x}
%  \equiv

%  \;,
% \ee
% which can be arranged for any $\mu$ and $\nu$, with $\hat C$ as large as desired, by choosing $\delta $ small enough and $x$ large enough.

We have now shown that, under the present hypotheses, we can construct sub- and supersolutions $1\pm\hat{C} e^{-\delta x}$ of
\eq{pL2} outside a sufficiently large compact set in $M$, with $\hat C$ as large as desired. To produce global
barriers, we must argue a bit further.  Recall that there is a conformally related metric
$\bar{g} =\zeta^{\frac{4}{n-2}} \tg$, with $\zeta\to1$ and $\partial \zeta,\partial^2 \zeta$ going to zero exponentially fast,
which has $\bar{R} \leq -c < 0$ on all of $M$. So if we
write the Lichnerowicz equation relative to this metric, then sufficiently small and large positive
constants $\eta$ and $\eta^{-1}$ are sub- and supersolutions. This means that
$\eta \zeta$ and $\eta^{-1} \zeta$ are global sub- and supersolutions for the equation relative to $\hg$,
and hence $\eta \phi \zeta$ and $\eta^{-1} \phi \zeta$ are global sub- and supersolutions for the equation relative to $\tg$.

All this leads us eventually to the global weak sub- and supersolutions
$$
\max \{ \phi( 1-\hat C  e^{-\delta x}) ,  \eta \phi \zeta \}   \ \mbox{and} \ \min \{\phi(1+\hat C  e^{-\delta x}), \eta^{-1} \phi\zeta \}
$$
for the original equation relative to $\tg$; the lower bound $\mu$ which appears in the argument above should be determined by
$\mu := \inf  \eta   \zeta >0$. Notice that these functions converge exponentially quickly to one another as $x \to \infty$,
which controls as well the asymptotics of the solution which lies between them.

We have now proved:
\begin{Theorem}
 \label{T11XII0.1}
Fix any number $c > 0$ and suppose that $\beta\ge c$. Let $\tg$ be a metric with a finite number of asymptotically
conformally cylindrical ends on $M^n$, $n\ge 3$, such that the limiting scalar curvature $\ztR =\lim_{x\to\infty}\tR $
is strictly negative in each end.  Assume also that $\tg$, $\tR$, $\tsigma$ and $ \beta$ all approach their
limits on each end at the rate $e^{-\nu x}$.  Then there exists a solution $\tphi $ of the Lichnerowicz
equation so that $\tphi$ converges (at some exponential rate) to a limiting function $\zphi$, which
is a solution of \eqref{redlich}. In particular, the solution metric $\tphi^{\frac{4}{n-2}}\tg$ remains
asymptotically conformally cylindrical.
\end{Theorem}

\subsubsection{Asymptotically periodic ends}
 \label{ss2I1.1}
As described at the beginning of this section, the metric $\tg$ is said to be conformally asymptotically periodic
on an end $E_\ell$ if, on that end, $\tg = w^{\frac{4}{n-2}}g$ where $g$ differs from
a metric $\check{g}$ which has period $T$ on $\RR \times N$, where the difference is a tensor which has local
H\"older norm decaying like $e^{-\nu x}$:
\bel{2I1.1}
 h:= g - \check{g}  \ \, \mbox{satisfies}\
 \| e^{ \nu x} h\| _{2+\ell,\alpha} < \infty \ \mbox{for some}\ \nu > 0
 \;.
\ee
The conformal factor $w$ is assumed to decay to a smooth periodic function
at the same exponential rate.  We also suppose that both $\sigma$
and $\beta$ decay like $e^{-\nu x}$ to functions $\check{\sigma}$ and $\check{\beta}$
which both  have the same period $T$.

All such ends are cylindrically bounded, so Proposition~\ref{P19XII0.1}
already guarantees the existence of solutions $\tphi$ of the Lichnerowicz
equation in this context if we assume that the scalar curvature is
negative and bounded away from zero outside a compact set.

One special feature of this case is that we can exploit known results on compact
manifolds.  Thus suppose that $\tg$ is asymptotic to the periodic metric $\cg$.
This induces a metric, which we give the same name, on $S^1_T \times N$
(here $S^1_T = \RR / T \mathbb Z$).  We can invoke the complete existence
theory for the Yamabe problem on compact manifolds to choose a
smooth positive conformal factor $\lambda$ on $S^1 \times N$
such that $\lambda^{\frac{4}{n-2}} \cg$ has constant scalar curvature.
This metric lifts to a periodic metric on $\RR \times N$, and we can obviously
replace $\cg$ by this new conformally related metric so as to assume
that $\tg \sim \check{w}^{\frac{4}{n-2}} \cg$ where $\cg$ has constant scalar curvature.
Note that even if $\cg$ is asymptotically cylindrical, i.e.\
$\cg = dx^2 + \zg$, the modified metric may be only asymptotically
periodic.  The constraint Delaunay metrics in \S\ref{sS10XI.1} are an example of this.

We now show that if the ends of $(M,\tg)$ are all asymptotically
conformally periodic, then there is a solution of the Lichnerowicz
equation which has the same structure.   To prove this, let $\cphi\in
\calC^{2,\alpha} (S^1\times N)$ be the (strictly positive) solution
of the limiting equation
\bel{2I1.4}
L_{\cg} \cphi = - \csigma ^2 \cphi^{-\gamma} + \cbeta\cphi^{\alpha}.
\ee
Assuming that $\cR \leq -c < 0$, existence is easy: the subsolution
is a small positive constant $\eta$ and the supersolution is a large positive constant $\eta^{-1}$.

Modifying the metric as above, we now search for solutions of
the form $\tphi = \phi + \psi$, where $\psi$ decays like $e^{-\nu x}$.
It is then possible to follow the steps in the argument of \S~\ref{ss2I1.5}
almost exactly. This leads to the

\begin{Theorem}
 \label{T2I1.1}
Suppose that $\beta \ge c > 0$, and let $\tg$ be a metric with a finite number of conformally asymptotically
periodic ends on $M^n$, $n \geq 3$. Assume that $\tR \leq -c < 0$
outside a compact set. Then there exists a solution $\tphi $ of the Lichnerowicz equation
so that each end remains asymptotically conformally periodic.
\end{Theorem}

There is an obvious extension of Theorems~\ref{T11XII0.1} and \ref{T2I1.1} which
allows $(M, \tg)$ to have ends of both types. We leave details to the reader.

\subsection{$\tR\ge 0$}
 \label{ss26XI0.1}
We now consider metrics on a manifold with a finite number of conformally asymptotically
cylindrical or periodic ends, assuming that $\tR$ enjoys certain positivity properties.

\subsubsection{Yamabe positivity}
 \label{ss21XII11}
We start with a discussion of the Yamabe invariant in the present context. The results here are only intended
to provide an alternative perspective to the results of \S \ref{ss21XII11.2}, and the analysis and results
there are independent of the ones here.

Our initial result corresponds to the first step in the proof of Proposition~\ref{P19XII0.1}.  In that setting,
if $\tR \leq -c < 0$ outside a compact set, then without any further conditions it is possible to make an
initial conformal change to arrange that $\tR \leq -c < 0$ on all of $M$.  The analogue
of this result when $\tR \geq 0$ requires an extra hypothesis on the Yamabe invariant, as defined
in \eqref{25X0.1x}.

It is well known that if $M$ is compact, then $-\infty < Y(M,[\tg]) \leq Y(S^n)$, and that the infimum
of the functional $Q_{\tg}$ is always realized. For manifolds with exactly cylindrical ends, a parallel though
somewhat less complete theory is developed in \cite{Ak-Bot}. Suppose that the ends are modelled on the products
$(\RR^+ \times N_\ell, dx^2 + \zg_\ell)$. As before, we assume for simplicity that there is only
one end and drop the subscript $\ell$. Although \cite{Ak-Bot} assumes that the restriction of
$\tg$ on the end is exactly of product type, it is not hard to extend their results to the case
of asymptotically cylindrical metrics, see \cite{AkCarrMaz}. Furthermore, if $\tg$ is conformally asymptotically
cylindrical, so $\tg = w^{4/(n-2)}\check{g}$, where $\check{g}$ is asymptotically cylindrical, then it
is obvious that $Y(M, [\tg]) = Y(M, [\check{g}])$. Thus we may talk about the Yamabe invariant
in this slightly more general setting.

One of the key points in \cite{Ak-Bot} is that $Y(M, [\tg])$ is regulated by the lowest eigenvalue $\mu_0$
of the operator
\begin{equation}
-\calL_{\zg} := -\Delta_{\zg} + c(n) \zR,
\label{redcL}
\end{equation}
which we have already encountered as the linear part of the reduced Lichnerowicz operator \eqref{redlich}.
(Just as for that operator, this is \emph{not} the conformal Laplacian on $N$ because the constant $c(n)$ is
appropriate for an $n$-dimensional manifold, not an $(n-1)$-dimensional one.)
The first main result \cite[Lemma 2.9]{Ak-Bot} is that $\mu_0 \geq 0$ if and only if $Y(M, [\tg]) > -\infty$.
It is not hard to see that
\begin{equation}
Y(M,[\tg]) \leq Y( \RR \times N, [dx^2 + \zg]) \leq Y(S^n),
\label{modifiedAubin}
\end{equation}
with the second equality sharp except when $(M,\tg)$ is conformal to the sphere, with its standard metric,
punctured at a finite number of points. If $Y(M,[\tg]) < Y(\RR \times N, [dt^2 + \zg])$ and $\mu_0 > 0$,
then there exists a metric in $[\tg]$ which minimizes the Yamabe functional, but this metric is incomplete,
with isolated conic singularities and the minimizing function $u$ decays exponentially along the ends.
(If $\mu_0 = 0$ the existence theory is more complicated, and the solutions when they exist are complete but of cusp type.)
The paper \cite{AkCarrMaz} contains a substantial generalization of their result and a closer investigation of
the asymptotics of the solution, even in the cylindrical setting.

Before stating our first result, we recall an alternate characterization of the positivity of the Yamabe invariant which
is needed in the argument.

\begin{lemma}
\label{Yamspec}
Let $(M,\tg)$ be a complete manifold with ends which are conformally asymptotically cylindrical or periodic.
Then the condition $Y(M, [\tg]) > 0$ is equivalent to the pair of conditions
\begin{eqnarray}
\mathrm{i)} &\mbox{the $L^2$ spectrum of $- L_{\tg}$ lies in $[0,\infty)$, and}
 \label{specconds1} \\
\mathrm{ii)} &\mbox{\ $0$ does not lie in the point spectrum of $-L_{\tg}$. }
\label{specconds2}
\end{eqnarray}
\end{lemma}

Before starting the proof, we remark that unlike Lemma~\ref{sgnYam}, where we only considered the signs of
these invariants on a compact manifold with boundary, here we consider their behaviour
on the noncompact manifold $M$.

\medskip

\begin{proof}
Denote by $\lambda_0(K)$ the lowest eigenvalue of $-L_{\tg}$ with Dirichlet boundary conditions
on any smoothly bounded compact subdomain of $M$, and $Y(K, [\tg])$ the Yamabe invariant
of any such subdomain. By Lemma~\ref{L19XII11.1}, these have the same sign for any $K$.
By (strict) domain monotonicity, if $\lambda_0(M) \geq 0$ then $\lambda_0(K) > 0$ for any $K$,
hence $Y(K, [\tg]) > 0$ and finally $Y(M, [\tg]) \geq 0$. Similar reasoning shows that
if $Y(M, [\tg]) > 0$ then $\lambda_0(M) \geq 0$.

It suffices now to prove that $0$ is an $L^2$ eigenvalue if and only if $Y(M, [\tg]) = 0$.
Suppose first that $0$ is an $L^2$ eigenvalue, and let $u_0$ be the corresponding
eigenfunction. Elliptic regularity implies that $u_0$ is as differentiable as the metric allows, and  completeness of $M$ together with the usual sequence of cutoffs  shows that $u_0 \in H^1(M)$, with $u_0 > 0$. Since the ends of $(M,\tg)$ are cylindrically
bounded, we claim that $(M, \tg)$ admits a Sobolev inequality (we substantiate this momentarily),
and hence $u_0 \in L^{2n/(n-2)}$ as well. It is then straightforward to check that cutting off $u_0$ to
be supported in larger and larger sets gives a minimizing sequence for $Q_{\tg}$, so that $Y(M, [\tg]) = 0$.
In fact, $u_0$ itself is a minimizer for this functional and $Q_{\tg}(u_0) =0$.

Regarding the claim about the existence of a Sobolev inequality, we recall the following standard facts
(discussed in more detail in \cite{AkCarrMaz}).  First, the existence of a Sobolev inequality is a
quasi-isometry invariant, so we may as well assume that $(M, \tg)$ has exact cylindrical ends.
Next, such an inequality is localizable, so it is enough to show that any cylinder $\RR \times N$,
where $N$ is compact, admits a Sobolev inequality. Finally, if a Sobolev inequality holds
for $N_1$ and $N_2$, then it holds for the Riemannian product $N_1 \times N_2$. Thus
since any compact manifold $N$ admits a Sobolev inequality, as does the real line, we see
that the same is true for the cylinder, and hence the manifold $M$.

Suppose conversely that $Y(M, [\tg]) = 0$. We now quote \cite[Proposition 1.1]{AkCarrMaz},
which guarantees in this setting that $Y(\RR \times N, [dx^2 + \zg]) > 0$ because $(M, [\tg])$
admits a Sobolev inequality.  But now we may apply \cite[Theorem A]{Ak-Bot}, which asserts
that since $Y(M, [\tg]) < Y( \RR \times N, [dx^2 + \zg])$, there exists a minimizer $u$ for
$Q_{\tg}$, so that $Q_{\tg}(u) = 0$.  Theorem B of that same paper asserts that $u$ decays
exponentially, and thus lies in $L^2$. Since the Euler-Lagrange equation for $Q_{\tg}$ for this function $u$ is
simply $-L_{\tg} u = 0$, we see that $0$ is an $L^2$ eigenvalue.
\qed
\end{proof}

\begin{proposition}
Suppose that $(M,\tg)$ has conformally asymptotically cylindrical or periodic ends, and that $Y(M, [\tg]) > 0$.
Then there exists $u \in \calC^\infty(M)$ with $u > 0$ on all of $M$ and $u \to u_0 > 0$ as $x \to \infty$,
where $u_0 \in \calC^\infty(N)$ in the cylindrical case and $u_0 \in \calC^\infty(S^1 \times N)$ in the periodic case,
such that $\check{g} = u^{4/(n-2)}\tg$ has  $\check{R} > 0$ everywhere.
\label{Rg0}
\end{proposition}

\proof   We proceed by modifying the metric in two steps: in the first, we change conformally to a metric which has
strictly positive scalar curvature outside a compact set, and in the second we arrange that the scalar curvature
becomes positive everywhere. As noted earlier, we may as well replace the metric $\tg$ by the conformally
related one $\check{g}$ which is asymptotic to $dx^2 + \zg$ in the asymptotically cylindrical setting and
to $\zg$ in the asymptotically periodic setting. We discuss the former case first.

For the first step, by our hypothesis and \eqref{modifiedAubin}, we see that $Y( \RR \times N, [dx^2 + \zg]) > 0$.
By \cite[Lemma 2.9]{Ak-Bot}, this implies that $\mu_0 > 0$, where as above, $\mu_0$ is the lowest
eigenvalue of $-\calL_{\zg}$. Let $\psi_0$ be the ground state eigenfunction for $\mathcal L_{\zg}$, so that
$\mathcal L_{\zg} \psi_0 = -\mu_0 \psi_0$ and $\psi_0 > 0$ on $N$.

Now choose a function $\psi$ which is strictly positive on all of $M$ and which equals $\psi_0$ far out
on each end.  Define $\overline{g} = \psi^{\frac{4}{n-2}} \check{g}$. Then
\[
L_{\check{g}} \psi = - c(n) \overline{R} \psi^{\alpha}.
\]
Comparing this with the fact that $L_{\check{g}} \psi \to \mathcal L_{\zg} \psi_0 = -\mu_0 \psi_0$
as $x \to \infty$, we see that $\overline{R} \geq c > 0$ outside a compact set.

For simplicity, rename the new metric as $\tg$ again, so that we now prove our statement assuming
that $\tR \geq c > 0$ far out on the ends.

For the next step, choose a smooth function $f$ which equals $- c(n) \tR$ outside a large compact set and
such that $f \leq -c < 0$ everywhere. We claim that there is a unique bounded solution to the equation
$L_\tg u = f$.  Setting $u = 1 + v$, we can rewrite this equation as
\[
L_{\tg} v = f + c(n) \tR := \tilde{f}.
\]
By construction $\tilde{f}$ is compactly supported.

We construct barrier functions for this equation.  By Lemma~\ref{Yamspec},
the spectrum of $-L_{\tg}$ is nonnegative and $0$ is not an $L^2$ eigenvalue.
Since we have arranged that $\tR \geq c > 0$
on the ends, the continuous spectrum of $-L_{\tg}$ must lie in a ray $[c', \infty)$ for some $c' > 0$, hence the interval
$(0,c')$ contains only isolated $L^2$ eigenvalues of finite rank. Putting these facts together, we see that $0$ does
not lie in the closure of the spectrum.

To proceed, assume first that the bottom of the spectrum of $-L_{\tg}$ is an isolated $L^2$ eigenvalue $\lambda_0 > 0$.
Thus $\lambda_0$ is necessarily simple and there exists a strictly positive $L^2$ eigenfunction $\phi_0$ on $M$ such
that $-L_{\tg} \phi_0 = \lambda_0 \phi_0$.   Now, since $\tilde{f}$ is compactly supported and $\phi_0 > 0$
everywhere, we can choose $A \gg 0$ so that
\[
L_{\tg} A \phi_0 = - A \lambda_0 \phi_0 \leq \tilde{f}, \quad \mbox{and}\quad L_{\tg} (-A \phi_0) = A \lambda_0 \phi_0
\geq \tilde{f}.
\]
By Proposition~\ref{P8XII11.1} in Appendix~\ref{A16XII11.1}, this gives a solution $v$ to $L_{\tg} v = \tilde{f}$
with $-A \phi_0 \leq v \leq A \phi_0$. Since $\tg$ is conformally asymptotically cylindrical or periodic, it
is known that $\phi_0$ decays exponentially (see \cite{AkCarrMaz} for the cylindrical case and
\cite{MPU1} for the periodic case).

If there is no $L^2$ eigenvalue below the continuous spectrum, we can still obtain sub- and supersolutions as
follows.  Fix any $0 < \lambda < \inf \mbox{spec}\,(-L_{\tg})$. A theorem due to Sullivan \cite[Theorem 2.1]{Sullivan2}
gives the existence of a strictly positive function $\phi$ such that $L_{\tg} \phi = -\lambda \phi$. Strictly
speaking, Sullivan's theorem is only stated for the Laplace operator itself rather than for an operator
of the form $\Delta_{\tg} - c(n) \tR$. However,  inspection of the proof \cite[p.~337]{Sullivan2} shows
that the proof adapts immediately to this slightly more general setting since it relies only on the Harnack inequality, the
existence of which for $L_{\tg}$ is classical \cite{GT}.

Choose a large compact set $K \subset M$ so that $\tR \geq c > 0$ outside $K$. Then
$L_{\tg} e^{-\epsilon x} \leq -c' e^{-\epsilon x}$ outside of $K$.  Choose a slightly smaller compact set
$K' \subset K$ and a constant $\eta > 0$ so that $\eta \phi > e^{-\epsilon x}$ on $\del K'$ but
$\eta \phi < e^{-\epsilon x}$ on $\del K$, and then define $\tilde{\phi} = \min\{ \eta \phi, e^{-\epsilon x}\}$
in $K \setminus K'$, $\tilde{\phi} = \eta \phi$ in $K'$ and $\tilde{\phi} = e^{-\epsilon x}$ outside of $K$.
This is the minimum of two supersolutions in the overlap region $K \setminus K'$ and equal to a smooth
supersolution outside of the overlap, so is a global weak supersolution (compare Appendix~\ref{A16XII11.1}). It is also strictly positive and uniformly bounded.
Thus, as before, sufficiently large multiples of $-\tilde{\phi}$ and $\tilde{\phi}$ serve as barriers to find a solution
$v$ which, as before, decays exponentially.

To conclude the argument, the function $u = 1 + v$ solves $L_{\tg} u = f < 0$.  We know that $u > 0$ outside
a large compact set. To show that $u > 0$ everywhere, choose a
smoothly bounded compact set $K$ so that $\tR > 0$ outside $K$, and let $\lambda_0(K)$ be the lowest eigenvalue
of $-L_{\tg}$ with Dirichlet boundary conditions on $K$. If $K$ is sufficiently large, $\lambda_0(K) > 0$
and the associated eigenfunction $\phi_0$ is strictly positive.  A brief calculation shows that
\[
\Delta_{\tg} \frac{u}{\phi_0} + 2 \frac{\nabla \phi_0}{\phi_0} \cdot \nabla (\frac{u}{\phi_0}) -
\lambda_0(K) \frac{u}{\phi_0} = \frac{f}{\phi_0} < 0\;.
\]
This equation shows that $u/\phi_0$ cannot attain a nonpositive minimum at any point $q \in K$,
and since $u/\phi_0 \to +\infty$ at $\partial K$, we conclude that $u > 0$ everywhere.

The modifications of this argument to the conformally asymptotically periodic setting are straightforward
and left to the reader.
\qed
%\end{proof}

\subsubsection{The Lichnerowicz equation}
 \label{ss21XII11.2}

Let us now return to the main problem.  We assume that $(M, \tg)$ has conformally asymptotically cylindrical
or periodic ends, and $0 \leq \beta \leq C$. We also assume that one of the following two hypotheses holds:
\begin{itemize}
\item[a)] $Y(M, [\tg]) > 0$, so by Proposition~\ref{Rg0}, we may as well assume that $0 < c \leq \tR \leq C$
on all of $M$, or
\item[b)] $0\leq \tR \leq C$ on all of $M$ and $\tR + \beta \geq c > 0$ outside some compact set.
\end{itemize}

We shall use an argument of Maxwell~\cite{MaxwellNonCMC} (see \cite{Hebeyslides} for a similar idea) which
we explain first in the following setting. Let $h$, $a$, and $f$ be smooth functions on a compact Riemannian
manifold $(N,\zg)$, with $f\ge 0$, $a \ge 0$ and $f+h\geq c > 0$. Consider the equation
%X
\begin{equation}
 \label{15X0.1} \Delta_{\zg} u - hu = fu^{\alpha} - a u^{-\gamma}
% \;,
\end{equation}
(the specific values of $\alpha>1$ and $\gamma>0$ are unimportant here). We also require that $a \not \equiv 0$.

By the strict positivity of $h+f$ (and the compactness of $N$), there exists a unique $u_{1 }$ such that
\bel{6XII0.1}
 \Delta_\zg u_{1 } - (h+ f) u_{1 } = -a
 \;.
\ee
This function is strictly positive by the maximum principle.  If $t > 0$ is sufficiently small, then
$u_t = tu_{1 }$ is a subsolution of \eq{15X0.1}: indeed, $ta \le   a t^{-\gamma}u^{-\gamma}_{1}$ and
$ f tu_1 \ge ft^\alpha u^{\alpha}_{1}$ for $t \ll 1$, thus
\bel{8I1.1}
 \Delta_\zg u_{t} - hu_{t} = -ta + t  u_1 f \ge  - a t^{-\gamma}u^{-\gamma}_{1}
 +ft^\alpha u^{\alpha}_{1}
 \;.
\ee
On the other hand, a large positive constant provides a supersolution. Hence \eqref{15X0.1} has
a solution.

In our particular setting, if $\zbeta\ge 0$, $\ztR \geq 0$ and $\zbeta + \ztR \geq c > 0$, and
in addition, $\ztsigma^ 2\not \equiv 0$, then we immediately obtain a solution $\zphi$ of the
limiting equation \eqref{redlich} by this argument.

Now we adapt this same construction to a manifold $(M,\tg)$ with conformally asymptotically
cylindrical ends. The analogue of \eq{6XII0.1} is now
\bel{6XII0.2}
\Delta_\tg u_{1 } - (c(n) \tR + \beta) u_{1 } = - \tsigma^2
\;.
\ee
We assume that $\tR$, $\beta$ and $\tsigma^2$ converge to their asymptotic limits $\ztR$, $\zbeta$ and $\ztsigma^2 \not\equiv 0$
at an exponential rate $e^{-\nu x}$ in the sense that
$\tR - \ztR \in e^{-\nu x} \mathcal C^{0,\mu}(M)$  for some $\mu \in (0,1)$, and similarly for the other two functions.
Here $\calC^{0,\mu}(M)$ consists of all functions $u \in \calC^{0,\mu}_{\mathrm{loc}}(M)$ such that
\[
\sup_{X \geq 0}  ||u||_{0,\mu, [X, X+1] \times N} < \infty,
\]
where $|| \cdot ||_{0,\mu, [X, X+1]\times N}$ denotes the H\"older norm on the finite cylindrical section
where $X \leq x \leq X+1$.

As above, fix a solution $\zu$ of the limiting equation \eqref{redlich}.
% %
% \bel{6XII0.3}
%  \Delta_\zg \zu_{1 } - (\zh+ \zf) \zu_{1 } = -\za
%  \;,
% \ee
% %
% for a function $\zu_1$ on $N$ which is strictly positive by the
% maximum principle.
In the asymptotically periodic case we make the obvious corresponding hypotheses, with $\ztR$, $\zbeta$
and $\ztsigma^2$ periodic, all with the same period as the metric, and we can then choose a periodic solution $\zu$.

Choose a function $\zu_1$ which coincides with $\zu$ far out on each end. We seek a solution of
\eq{6XII0.2} of the form $u_1 = \zu_1 +v$. Thus $v$ must satisfy
\bel{6XII0.5}
\Delta_\tg v - (c(n) \tR + \beta) v = - \tsigma^2  -  (\Delta_\tg \zu_1  - (c(n) \tR + \beta) \zu_1 )
 \;.
\ee
One readily checks that in either the asymptotically conformally cylindrical or periodic
settings, the right-hand side lies in $e^{-\nu x} \mathcal C^{0,\mu}(M)$.  Using the functions
\[
\phi = \min\{1, Ae^{- \nu' x}\}
\]
and $-\phi$ as super- and subsolutions, where $0 < \nu' < \nu$ and $A \gg 0$, we obtain a bounded
solution of this equation. Since $u_1 = \zu_1 + v  \to \zu > 0$ as $x \to \infty$, it must be strictly positive
outside a large compact set; the maximum principle applied to \eq{6XII0.2} gives $u_1 \geq c > 0$ everywhere.

Finally, the same calculation \eq{8I1.1} shows that $u_-:=tu_1$ is a strictly
positive subsolution of \eq{6XII0.2} when $0 < t \ll 1$, while a large constant
$u_+ := C$ provides a supersolution.

Calculations similar to  those of Section~\ref{ss2I1.5} (in fact simpler because for
$\tR \geq c > 0$  no careful study of the right hand side is needed)
% %
% %
% \footnote{The argument is in fact simpler if $\tR\ge
% \epsilon>0$, as then there is no need to borrow the term
% $\partial_\tphi G_1(\beta,\zphi)$ from the non-linearity to
% obtain a well behaved operator at the left-hand-side. If $\tR$
% has zeros, the argument given there can be repeated because of
% our assumption $\tR+\beta\ge\epsilon$.}
%
show that $\zu\pm C e^{-\delta x}$ are also sub- and supersolutions
for a sufficiently small $\delta>0$. Hence
\bel{15II11.10}
 \tphi_-=\max( \zu-C e^{-\delta x}, u_-) \ \mbox{and} \
 \tphi_+=\min( \zu+C e^{-\delta x}, u_+)
\ee
are weak sub- and supersolutions of the Lichnerowicz equation.

Altogether, we have proved the
\begin{Theorem}
 \label{T8I1.1}
Let $(M,\tg)$ have a finite number of conformally asymptotically cylindrical and
periodic ends. Suppose that either of the hypotheses a) or b) above hold,
%that $\beta + \zR \geq c > 0$ far out on the ends,
and
$$
\beta -  \zbeta  = \calO(e^{-\nu x})
  \;, \qquad \tsigma^2 -  \ztsigma^2  = \calO(e^{-\nu x})
  \;, \qquad \ztsigma^2 \not \equiv 0\;,
$$
for some functions  $\ztsigma^2>0$, $\zbeta$ on $N$ in the
asymptotically cylindrical ends, and for some functions
$\ztsigma^2$, $\zbeta$ on $S^1 \times N$ in the
asymptotically periodic ends. Then there exists a solution
$\tphi $ of the Lichnerowicz equation so that each end of
cylindrical type  remains of the same type for the metric
$\tphi^{\frac{4}{n-2}}\tg$.
\end{Theorem}

\section{Manifolds with asymptotically hyperbolic and
cylindrical ends}
 \label{ss2I1.3}
The constructions in Section~\ref{ss2I1.2} generalize immediately to manifolds
which have a finite number of ends, some asymptotically hyperbolic, as defined
in~\cite{AndChDiss}, and the others cylindrically bounded.  The cylindrical
ends are handled as in the  Section~\ref{ss2I1.2}, while the asymptotically
hyperbolic ones are treated as in~\cite{AndChDiss}. The reader
should have no difficulties supplying the details to establish the following
\begin{Theorem}
\label{T21XII0.1}
Let $(M,\tg)$ be a Riemannian manifold of dimension $n\ge 3$  with a finite number of conformally
asymptotically cylindrical or periodic ends and a finite number of asymptotically hyperbolic ends.
Suppose that $\beta\ge c > 0$, and that in each asymptotically cylindrical end, the limiting
scalar curvature $\zR $ is negative. We assume moreover that $\beta$ approaches a constant and
$\tsigma^2$ approaches zero in the asymptotically hyperbolic ends with the usual rates as in
\cite{AndChDiss}, and in addition that on each cylindrical end
$$
\tsigma^2 -  \ztsigma^2  = O(e^{-\epsilon x})   \;,
%  \qquad \ztsigma^2 \not \equiv 0\;,
$$
for some bounded function  $\ztsigma^2$ on $N$.  Then there exists a solution $\tphi $ of the
Lichnerowicz equation so that in the conformally rescaled metric, each end has the same
type as for $\tg$.
\end{Theorem}

\section{Manifolds with asymptotically flat and  cylindrical ends, $\tR\geq 0$}
 \label{S22I12.1}

We now extend the analysis of \S \ref{ss26XI0.1} and suppose that $(M,\tg)$ is a complete manifold with a
finite number of ends, each one either cylindrical with $\tR \geq 0$, or else asymptotically Euclidean or
conical. As mentioned in the introduction, it is simpler to refer to this last case as only asymptotically
Euclidean, the conical case being understood. We only consider the problem with
$$
\Lambda = 0,
 \;,
$$
%,
and assume that $\beta\ge  0$, with both $\beta$ and  $\tsigma^2$ tending to zero in the asymptotically
Euclidean ends faster than $r^{-2- \epsilon}$ for some $ \epsilon >0$. We also assume that
$\tsigma^2 \to \zsigma^2\not \equiv 0$ and $\beta \to \zbeta$ exponentially fast along each asymptotically
conformally cylindrical (in which case $\zsigma^2 $ and $\zbeta$ are functions on $N$) or asymptotically
periodic end (in which case $\zsigma^2$ and $\zbeta$ are periodic on $\RR \times N$, with the same period as
the metric).  Finally, we assume that sufficiently far out in the cylindrical ends the  scalar curvature is
bounded away from zero.

Special cases of the construction in this section have been considered
in~\cite{GabachClement,Waxenegger,GabachDain1,GabachDain2}.
The reader is referred to~\cite{ChIY,FriedrichYamabe,christodoulou:murchadha,IMP2} and references therein for more information on the constraint equations on asymptotically flat manifolds.

We start by arranging that $\tg$ has controlled curvature in the asymptotically Euclidean
regions.  For some constant $r_0 \gg 0$ to be chosen below, let $\chi$ be a smooth nonnegative
function which equals one for $r\ge r_0$ in the asymptotically Euclidean ends and vanishes for
$r \leq r_0 - 1$ and elsewhere on $M$. Let $\mu \in \calC^\infty$ be a small, nonnegative function with
support in $\{r \leq 3r_0\}$ such that $\mu + \tR >0 $ when $\{r \leq 2r_0\}$. Now, let $v$ be a
solution of the equation
\begin{equation}
L_{\tg} v + \chi\tsigma^2 +\mu=0
\label{22I1.2}
\end{equation}
such that $v \to 0$ on all the ends. We assume that on each such end
$$
\tsigma^2 \in  \rho_e^{-\nu-2} \cC^{0,\alpha}_\tg
$$
for some $-\nu \in (2-n, 0)$ and that
$$
\tsigma^2 - \zsigma^2 \in \rho_c^{-\nu}\cC^{0,\alpha} \;,  \quad \zsigma^2 \not \equiv 0
$$
on each conformally asymptotically cylindrical end, with $\nu$ sufficiently small as determined by the dimension and by the
asymptotic behavior of $\tR$ in the cylindrical ends. It is proved in Appendix~\ref{A16XII11.1} that with these
decay conditions, and since $\tR \geq 0$, a solution $v$ exists.

% %
% \[
% L_{\tg} =  \Delta_\tg - \frac{n-2}{4(n-1)}\tR:
%  \  \rho_e^{-\nu} \cC^{2,\alpha}_\tg(M) \longrightarrow \rho_e^{-\nu-2} \cC^{0,\alpha}_\tg(M)
% \]
% %
% is an isomorphism for any weight $2-n < -\nu < 0$ on each asymptotically Euclidean/conic end.

It is straightforward from this method of proof that $v$ tends uniformly to zero as
$r_0 \nearrow \infty$  and $\mu$ is made smaller. Thus we may assume that
$$
1+v \ge \frac 12 \;.
$$
The scalar curvature of the metric $\overline{g} = (1+v)^{\frac 4 {n-2}} \tg $ is equal to
\bel{22I11.3}
 %\lefteqn
\overline{R} = \frac{-L_{\tg} (1+v)}{c(n) (1+v)^{\frac {n+2}{n-2}}} =
\frac {\tR}  {(1+v)^{\frac {n+2}{n-2}}}+
 \frac { (\chi \tsigma^2 +\mu)}{c(n) (1+v)^{\frac {n+2}{n-2}}}
% &&
%\\
 \;.
\ee
{}From this we see that $\overline{R} > 0$ everywhere and $\overline{R} \geq c \tsigma^2$ for
$r \geq r_0$ in the asymptotically Euclidean and conical regions.
For simplicity, we write this new metric $\overline{g}$ as $\tg$ again.

% As in \S~\ref{ss26XI0.1} in each conformally asymptotically cylindrical end we solve the limiting equation
% %
% \begin{equation}
% \Delta_\zg \mathring \tphi - \frac {n-2}{4(n-1)}\ztR \mathring \tphi   = - \ztsigma ^2 \mathring \tphi^{-\gamma}
%  +  \zbeta \mathring \tphi^{\alpha},
% \label{22I11.1}
% \end{equation}
% %
% and similarly in the asymptotically periodic ends.
%  \ptcr{the last paragraph is repeated in the new two paragraphs, but with a different notation}

We now adapt to this setting the construction of barrier functions from \S\ref{ss26XI0.1}.
The hypotheses in the cylindrical ends are identical to those in that section. Consider again
the equation \eq{6XII0.2},
\bel{6XII0.2x}
\Delta_\tg u_{1 } - (h+ f) u_{1 } = -a \;.
\ee
In each asymptotically conformally cylindrical end we solve the limiting equation
\bel{6XII0.3x}
\Delta_\zg \zu_{1 } - (\zh+ \zf) \zu_{1 } = -\za ;
\ee
the function $\zu_1$ on $N$ is strictly positive by the maximum principle. Similarly, on each
asymptotically periodic case we obtain a positive periodic solution $\zu_1$.

Now let $\zu$ denote any function on $M$ which coincides with $\zu_1$ far out in each cylindrical end
and which equals one on each asymptotically Euclidean or conical end.  We search for
a solution $u_1$ of \eq{6XII0.2x} by setting $u_1 = \zu +v$. The function $v$ must satisfy
\bel{6XII0.5x}
(\Delta_\tg - (h+ f)) v = -a -  (\Delta_\tg - (h+ f)) \zu \;.
\ee
This has suitable behavior in each asymptotic ends so that we may use the barrier functions described
in Appendix~\ref{A16XII11.1}. We obtain a solution $v$ which satisfies  $v \in e^{-\nu_1 x} \calC^{2,\alpha}(M)$
in the cylindrical ends, for some $\nu_1 \in (0,\nu)$, and such that $v$ decays faster than $ r^{-\epsilon}$
on the asymptotically Euclidean and conical ends.

The solution $u_1 = \zu +v$ of \eq{6XII0.2} tends to the strictly positive function $\zu$, hence $u_1$
must be strictly positive far out on the cylindrical ends. The maximum principle applied to \eq{6XII0.2x}
shows that $u_1 \geq c > 0$. The calculation \eq{8I1.1} shows that $u_-:=tu_1$ is a strictly
positive subsolution of \eq{6XII0.2} when $t > 0$ is sufficiently small.

On the other hand, a large constant $u_+ := C$ provides a supersolution. This is clear for any domain $K$
in $M$ which does not intersect the asymptotically Euclidean or conical regions since the scalar curvature
$\tR$ is strictly positive. On these other ends, we see that it is a supersolution using that
$\tR \ge c \tsigma^2$.

We now use the definition \eq{15II11.10} in the cylindrical ends and
\bel{15II11.10x}
 \tphi_-=\max( 1-C r^{-\delta }, u_-) \ \mbox{and} \
 \tphi_+=\min( 1+C r^{-\delta  }, u_+)
  \;.
\ee
on the asymptotically Euclidean and conic ends, for $\delta > 0$ sufficiently small.

All of this leads to the
\begin{Theorem}
 \label{T28XI0.1new}
Let $(M^n,\tg)$ be a complete Riemannian manifold, $n \geq 3$, with a finite number asymptotically Euclidean and conical
ends and a finite number of conformally asymptotically cylindrical and periodic ends.  Assume that $\tR\ge 0$
and $\tR \geq c > 0$ on the cylindrical ends. Then for any $\tsigma^2$ and $\beta\ge 0$ satisfying
$$
\tsigma^2\;,\ \beta \in \rho_e^{-\nu-2} \cC^{0,\alpha}
$$
in each asymptotically flat or conical end and
$$
 \tsigma^2 - \zsigma^2
 \;, \
 \beta -\zbeta
  \,
   \in \rho_c^{-\nu}\cC^{0,\alpha}_\tg
  \;, \qquad \ztsigma^2 \not \equiv 0\;,
$$
in each asymptotically conformally cylindrical end,  there
exists a solution $\tphi $ of the Lichnerowicz equation with
$\Lambda=0$ so that each end in the conformally rescaled metric
has the same asymptotic type.
\end{Theorem}

\appendix

\section{The barrier method for linear and semilinear elliptic equations}
 \label{A16XII11.1}
We review here, for the reader's convenience, two well-known results which are invoked repeatedly
in this paper to establish existence of solutions to the certain linear and semilinear elliptic equations
which arise in various geometric settings. As explained in the introduction, we rely entirely on barrier methods
in this paper rather than, for example, parametrix methods. While these barrier techniques rarely provide the
sharpest mapping properties or decay rates, they have several advantages; in particular, when they work,
they tend to be much simpler than other methods, and usually require much less regularity.

\medskip

\noindent{\bf Construction of barrier functions}
\begin{Proposition}
 \label{P8XII11.1}
Let $(M,\tg)$ be a smooth Riemannian manifold, and $h$ any nonnegative smooth function on $M$.
Suppose that $f$ is smooth and there exist two $\calC^0$ functions $\dtphi \leq \utphi$ which satisfy
\[
(\Delta_{\tg} - h) \dtphi \geq f, \qquad (\Delta_{\tg} - h) \utphi \leq f
\]
in the weak sense. Then there exists a smooth function $u$ such that
\begin{equation}
(\Delta_\tg -h) u = f\quad \mbox{and}\quad \dtphi \leq u \leq \utphi.
\label{5l1.1}
\end{equation}
\end{Proposition}
\begin{proof}
Choose an exhaustion of $M$ by a sequence of compact manifolds with smooth boundary $M_j$.
Because of the sign of $h$, the inhomogeneous Dirichlet problem
\[
(\Delta_{\tg} - h) u_j = f,\quad \left. u \right|_{\del M_j} = \left. \dtphi\right|_{\del M_j}
\]
is uniquely solvable for every $j$. By the standard (weak) comparison principle, $\dtphi \leq u_j \leq \tphi$ on $M_j$.

Letting $j \to \infty$, we see that the sequence  $\{u_j\}$ is uniformly bounded on every
compact set $K \subset M$. Using local elliptic estimates and the Arzela-Ascoli theorem on each $M_i$,
a diagonalisation argument shows that some subsequence $u_{j'}$ converges in $\calC^\infty$ on every
compact set. The limit function $u$ satisfies the correct equation and is sandwiched
between the two barriers $\dtphi$ and $\utphi$.

Note that there is no a priori reason for this solution to be unique, although this may be true in certain circumstances.
\qed\end{proof}

A useful aspect of this is that the barriers need only be continuous rather than $\calC^2$.
The simplest situation in which such a more relaxed hypothesis may arise is the following:
\begin{Lemma}
 \label{L9XII11.1}
Suppose that $\tphi_1$ and $\tphi_2$ are two subsolutions for the equation $(\Delta_{\tg} - h) u = f$. Then
$\tphi = \max\{ \tphi_1, \tphi_2\}$ is also a subsolution in the sense that if $u$ is a solution
to this equation on a domain $D$ and if $u \geq \tphi$ on $\del D$ then $u \geq \tphi$ on $D$.
Similarly if $\tphi_1$ and $\tphi_2$ are supersolutions, then $\tphi = \min \{ \tphi_1, \tphi_2\}$
is also a supersolution.
\end{Lemma}
\begin{proof}
Observe that $u \geq \tphi \geq \tphi_j$ on $\del D$, hence $u \geq \tphi_j$ on $D$. Since this is
true for $j = 1,2$, we have $u \geq \tphi$ on $D$ too. \qed
\end{proof}

In the applications encountered in this paper, one of the subsolutions, say $\tphi_2$, is typically only defined
on some open subset of $M$ rather than on the whole space, so the argument above does not quite work.
Thus we formulate this result in a slightly more general way.

As in the proof of Proposition~\ref{P8XII11.1}, it suffices to consider barriers on a compact manifold
with boundary, since when $M$ is noncompact we construct solutions on a exhaustion of $M$
by compact manifolds with boundary $M_j$, and then extract a convergent sequence using
Arzela-Ascoli.

Thus let $M$ be a compact manifold with boundary, and suppose that $\partial M = \del_1 M \cup \del_2 M$
is a union of two components (which may themselves decompose further).  Suppose that $\calU$ is a relatively
open set $M$ containing $\del_2 M$, but which has closure disjoint from $\del_1 M$.
Let $\phi_1$ be a subsolution for the operator $L = \Delta - h$ which is defined on all of $M$,
and $\phi_2$ a subsolution for $L$ which is only defined on $\mathcal U$. We assume that $\phi_1$ and $\phi_2$
are continuous and subsolutions in the weak sense.

Consider the open set $\mathcal V = \{ \phi_2 > \phi_1\}$, and let us suppose that
\[
\del_2 M \subset \calV \subset \overline{\calV} \subset \calU.
\]
Define the function
\[
\psi = \begin{cases}
\max \{ \phi_1, \phi_2\} \  & \mbox{in} \  \calU \\
\ \ \phi_1 \ & \mbox{in} \  M \setminus \calU.
\end{cases}
\]
\begin{Lemma}
This function $\psi$ is continuous and a weak global subsolution for $ L$ on $M$.
\end{Lemma}
\begin{proof} The continuity of  $\psi$ is clear from the fact that $\phi_1 > \phi_2$ in a neighbourhood of the `inner'
boundary of $\calU$, i.e.\ $\del \calU \setminus \del_2 M$.

Next, suppose that $u$ is a solution defined on all of $M$ and that $u \geq \psi$ on $\del M$.
Thus $u \geq \phi_1$ on $\del_1 M$ and $u \geq \phi_2$ on $\del_2 M$.  By asumption,
$\phi_2 \geq \phi_1$ on $\del_2 M$ so $u \geq \phi_1$ on all of $\del M$, hence
since $\phi_1$ is a subsolution, $u \geq \phi_1$ on all of $M$. In particular, $u \geq \phi_1$ on
the set $Y = \{ \phi_1 = \phi_2\}$.

The set $Y$ is compactly contained in $\calU$, and furthermore, $\del \overline{\calV} = Y \cup \del_2 M$.
This means that $u \geq \phi_2$ on $\del \overline{\calV}$, hence $u \geq \phi_2$ on all of $\calV$.
Putting these facts together yields that $u \geq \psi$ on all of $M$ . \qed
\end{proof}

\medskip

\noindent{\bf Examples of barrier functions}

Now suppose that $(M, \tg)$ is a complete manifold with a finite number of ends, each of one of the six types
described in the introduction. We illustrate how the situation above arises by describing standard types of
sub- and supersolutions for the problem
\begin{equation}
(\Delta_{\tg} - h) u = f
\label{sss}
\end{equation}
on each of these types of ends. Here $f$ is a $\calC^{0,\mu}$ function which satisfies certain weighted
decay conditions which are implicit in each case.  In each of these geometries, the end $E$
is a product $\RR^+ \times N$, where $N$ is a compact manifold, but there is a different asymptotic structure
each time.

\begin{enumerate}
\item {\it Asymptotically conic ends (this includes asymptotically Euclidean ends)}: Here $\tilde g$ approaches the conic
metric $g_c := dr^2 + r^2 h$ for some metric $h$ on $N$ in the following sense. Using a fixed coordinate system
$(y_1, \ldots, y_{n-1})$ on $N$, augmented by $r = y_0 \geq 1$, we assume that
\bel{5I1.2}
\tg_{ij} - (g_c)_{ij} = o(1)\;, \quad \del_k (\tg_{ij} - (g_c)_{ij}) = o(r^{-1}).
\ee
Then
$$
u_\pm= \pm C \|r^{\alpha+2} f\|_{L^\infty} r^{-\alpha}
$$
are sub- and supersolutions of \eqref{sss} when $r \geq r_0$ and $C \gg 0$, provided $\alpha \in (0, n-2)$.

\item
{\it Conformally compact (asymptotically hyperbolic) ends}: We now assume that
that $x \in (0, x_0]$ and that $\overline{g} = x^2 \tg$ has components approaching those of $dx^2 + \zg$ as $x \to 0$,
where $\zg$ is a Riemannian metric on $N$, and that the derivatives of the coordinate components of $\overline{g}$
are $o(x^{-1})$. Now, if $\nu \in (0, n-1)$,
$$
u_\pm= \pm \|x^{-\nu} f\|_{L^\infty} x^{\nu}
$$
are sub- and supersolutions of  \eqref{sss} when $x_0 \ll 1$ and $C$ is sufficiently large.

\item {\it Asymptotically cylindrical ends}: Assume that on $[x_0,\infty)\times N$ the metric components of $\tg$
and their first coordinate derivatives approach those of $dx^2 + \zg$, where $\zg$ is a Riemannian metric on $N$.
We emphasize that no decay rate is required. Assume moreover that
$$
h\ge \eta^2 > 0,
$$
for some constant $\eta$. Then, if $\nu \in (-\eta, \eta)$,  the functions
$$
u_\pm= \pm C  \|e^{\nu x} f\|_{L^\infty} e^{-\nu x}
$$
are sub- and supersolutions of \eqref{sss} for $x_0 \gg 1$.

\item {\it Cylindrically bounded ends}: Consider a metric $\tg$ on $[x_0,\infty)\times N$.
To obtain exponentially decaying sub- and supersolutions we need to ensure that
\bel{5I1.5}
(\Delta_\tg - h ) e^{-\nu x}   \le  -C  e^{-\nu x}
\ee
for some constant $C$. Now
\bel{6I1.1}
(\Delta_\tg - h ) e^{-\nu x}  = (-\nu \Delta_\tg x + \nu^2 |d x|_\tg - h)  e^{-\nu x} \;,
\ee
and so \eq{5I1.5} holds if
\bel{5I1.6x}
 h \ge C -\nu \Delta_\tg x + \nu^2 |d x|_\tg   \;.
\ee
For example, this holds when $|\nu|$ is small enough provided there exists a constant $\epsilon>0$ such that
$$
h\ge \epsilon\;, \quad \Delta_\tg x \le \epsilon^{-1}\;, \quad |dx|_\tg \le \epsilon^{-1}
 \;,
$$
In particular, this holds for any cylindrically bounded metric (including conformally asymptotically cylindrical and
conformally asymptotically periodic metrics) provided $h\ge \eta^2 > 0$.
\end{enumerate}

\medskip

\noindent{\bf An alternate construction of barriers}

These various sub- and supersolutions take constant values at the boundary $\{x_0\} \times N$ of $E$,
and have the extra property that the gradient of $\mp u_\pm$ points into $E$ at the boundary
(in the cylindrical cases, one must assume that $\nu>0$ for this to hold).
If this sort of normal derivative condition holds, then there is an alternate
proof that these can be used to construct weak barriers.

Suppose that $(M,\tg)$ is the union of a smooth compact manifold with boundary $M_0$
with a finite number of ends $E_\ell$. Suppose too that on each end $E_\ell$ there are sub-
and supersolutions $u_{\ell, -} < 0 < u_{\ell,+}$ of \eq{sss} which take constant values on
$\del M_0$, and such that $\mp \nabla u_{\ell, \pm}$ is nonvanishing along $\del E_\ell$
and points into $E_\ell$. Possibly multiplying the $u_{\ell, \pm}$ by large constants,
we assume that all $u_{\ell,\pm}$ take the same constant value $\alpha_\pm$ on $\del M_0$.

Let $u_0$ be the solution of \eq{sss} on $M_0$ with $u_0 = 0$ on $\del M_0$.
Choose a large constant $C$ so that $|\nabla u_{\ell, \pm}| \geq C$ on $\del M_0$
gradient of each of the $u_\pm$ on $\partial M_0$ is everywhere
larger than $\sup_{\del M_0} |\nabla u_ 0|$. Then the functions
$$
\phi_\pm = \begin{cases}
u_0 + C \alpha_\pm\ \ & \mbox{on}\ M_0 \\
C u_{\ell, \pm} & \mbox{on}\ E_\ell,
\end{cases}
$$
are weak sub- and supersolutions of \eq{sss} on the entire manifold $M$. Indeed, the choice of $C$ guarantees that
the distributional second derivatives of $\phi_\pm$ have the appropriate signs at $\del M_0$.

\medskip

\noindent{\bf Using barriers to construct solutions of equations}

We turn, finally, to a consideration of how these barrier functions can be used to solve semilinear elliptic equations.

\begin{Proposition}[Monotone iteration scheme]
\label{P13XIII0.1}
Let $(M,\tg)$ be a smooth Riemannian manifold and $F$ a locally Lipschitz
function. Suppose that $\dtphi \le \utphi$ are continuous functions which satisfy
$$
\Delta_\tg \dtphi \ge F(z,\dtphi)\;,  \qquad  \Delta_\tg \utphi \le F(z,\utphi)
$$
weakly. Then there exists a smooth function $\tphi$ on $M$ such that
$$
\Delta_\tg \tphi = F(z,\tphi)\;, \qquad \dtphi \le \tphi \le  \utphi \;.
$$
\end{Proposition}

As in the linear case, we do not assert that the solution is unique, and there are examples which show
that uniqueness may fail.

\medskip

\proof
When $M$ is compact, we proceed as follows.  Let $\underline{\alpha} = \inf_M \dtphi$ and
$\overline{\alpha} = \sup_M \utphi$. Rewrite the equation as
\[
(\Delta_{\tg} - A^2) \tphi = F_A(z,\tphi),
\]
where
\[
F_A(z, \tphi)  := F( z, \tphi) - A^2 \tphi.
\]
Choose $A$ so large that the function $F_A$ satisfies $\partial F_A(z, \mu )/\partial \tphi < 0$ for almost every $\mu \in
[\underline{\alpha}, \overline{\alpha}]$.

Now
%fix $\nu = 0$,
set $\tphi_0 = \underline{\tphi}$, and define the sequence of functions $\tphi_j$ by
\[
(\Delta_\tg - A^2) \tphi_{j+1} = F_A(z, \tphi_j).
\]
To see that this is well-defined for every $j$, note simply that $\Delta_{\tg} - A^2$ is invertible and
furthermore, by the maximum principle and induction,
\[
\underline{\tphi}=\tphi_0 \leq \tphi_1 \leq \tphi_2 \leq \ldots < \overline{\tphi},
\]
for all $j$, which implies that $F_A$ is monotone for the same constant $A$ (which depends only on $\underline{\tphi}$
and $\overline{\tphi}$). Even though $\tphi_0$ is only continuous, standard elliptic regularity
shows that $\tphi_1 \in \calC^{0,\alpha}$ and that $\tphi_j \in \calC^{2,\alpha}$ for $j \geq 2$.
%So, assume by induction that $\tphi_j$ takes values in $[\e, \e^{-1}]$, so that $\tphi_{j+1}$ is well defined.

We have produced a sequence which is monotone and uniformly bounded away from $0$ and $\infty$,
so it is straightforward to extract a subsequence which converges in $\calC^{2,\alpha}$ for some $0 < \alpha < 1$.
If $F \in \calC^\infty$, then the subsequence converges in $\calC^\infty$ too.

All of this works equally well if $M$ is a compact manifold with boundary. To be concrete, we require at each
stage that $\tphi_j = \underline{\tphi}$ on $\del M$ and we obtain a solution in the limit which satisfies
the same boundary conditions.

Now consider a general manifold $(M,\tg)$. As in Lemma \ref{L9XII11.1}, choose an exhaustion $M_j$ of $M$
by compact submanifolds with smooth boundary.  For each $j$, choose $A_j$ so large that
$\partial F(z, \mu )/\partial \tphi -A_j^2< 0$ on $M_j$ for almost every
$$
\mu \in [\inf_{M_j}\dtphi, \sup_{M_j}\utphi]
 \;.
$$
We may as well assume that $A_j$ is a nondecreasing sequence.
%
% {\blue By the argument just given,  on every $M_i$ we obtain a limiting
%  smooth function $\tphi_i=\lim_{j\to\infty} \tphi_{i,j } $,
%  solution of the Yamabe equation.

%  The Arzela-Ascoli theorem shows that from the sequence
%  $(\tphi_i)_{i\in\N}$ of solutions on $M_1$ we can  extract a sequence
%  $(\tphi_{i_{(1,j)}} )_{j\in \N}$ such that
%  $(\tphi_{i_{(1,j)}}|_{M_1})_{j\in \N}$ converges to a solution
%  $\psi_ 1$ on $M_1$. Then from the sequence $(i_{(1,j)})_{j\in
%  \N}$ we can extract a subsequence $(i_{(2,j)})_{j\in \N}$ so
%  that $\tphi_{i_{(2,j)}}|_{M_2}$ converges to a solution
%  $\psi_2$ on $M_2$. Note that $(\tphi_{i_{(2,j)}}|_{M_1})_{j\in
%  \N}$ is a subsequence of $(\tphi_{i_{(1,j)}}|_{M_1})_{j\in \N}$
%  and thus $\psi_2|_{M_1} = \psi_1$. Continuing in this way, one
%  obtains a sequence  $(\tphi_{i_{(k,\ell)}} )_{k,\ell\in \N}$ so
%  that $(\tphi_{i_{(k,k)}} )_{k\in \N}$ converges to a smooth
%  solution $\tphi$ of the Lichnerowicz equation satisfying
%  %
%  \[
%  \underline{\tphi} \leq \tphi \leq   \overline{\tphi}
%   \;.
%  \]}

Using the first part of the proof, for each $j$ we can solve the equation
\[
(\Delta_{\tg} - A_j^2) \tphi_j = F_{A_j}(z, \tphi_j), \quad \left. \tphi_j \right|_{\del M_j} = \underline{\tphi}
\]
Notice that by adding $A_j^2 \tphi_j$ to both sides, the functions $\tphi_j$ all satisfy the same equation
and are all trapped between the two fixed barrier functions $\underline{\tphi}$ and $\overline{\tphi}$,
albeit on an expanding sequence of domains.  Elliptic estimates for the fixed equation
$\Delta \tphi = F(z,\tphi)$ may now be used to obtain uniform a priori estimates for derivatives
of $\tphi_j$ on any fixed compact set.  From this we can use Arzela-Ascoli and a diagonalization
argument to find a subsequence which converges in $\calC^{2,\alpha}$ (or $\calC^\infty$)
on any compact set to a limit function which satisfies the equation and which lies between the
same two barrier functions.
\qed
%\end{proof}

\section{Some examples}
 \label{A5I1.2}

The flagship example of black holes with degenerate horizons is
provided by the \emph{Majumdar-Papapetrou} black holes, in
which the metric $^4 {}g$ and the electromagnetic potential $A$
take the form
\bea\label{ABHI.0} &  {}^4 {}g = -u^{-2}dt^2 +
u^2(dx^2+dy^2+dz^2)\,, &
\\
\label{I.0.1}
 &A = u^{-1}dt\,. &
\eea
The \emph{standard MP black holes} are obtained if the
coordinates $x^\mu$ of \eq{ABHI.0}--\eq{I.0.1} cover the range
$\R\times(\R^3\setminus\{\vec a_i\})$ for a finite set of
points $\vec a_i\in\R^3$, $i=1,\ldots,I$, with the function $u$
taking the form
\be \label{standard} u=1+\sum_{i=1}^I
\frac{m_i}{|\vec x - \vec a_i|}
 \,,
\ee
for some strictly positive constants $m_i$. Introducing radial
coordinates centered at a puncture $\vec a_i$, the   metric $g$
induced on the slices $t=\const$ by \eq{ABHI.0} is
\bel{5I1.11}
 g = \frac {m_i^2 }{r^2}(1 + O(r))\big(dr^2 + r^2 (d\theta^2 + \sin^2
 \theta \, d\varphi^2)\big)
 \;.
\ee
The new coordinate $x=-\ln r$ leads to a manifestly
asymptotically cylindrical metric:
\bel{5I1.11x}
 g =   {m_i^2} (1 + O(e^{-x}))\big(dx^2  +\underbrace{d\theta^2 + \sin^2
 \theta \, d\varphi^2}_{=:\zg}\big)
 \;.
\ee

In this example the slices $t=\const$ are totally geodesic, and
it follows from the scalar constraint equation (with Maxwell
sources) that the scalar curvature of $g$ is positive
everywhere. Furthermore  the scalar curvature of the metric
$\zh$ defined in \eq{5I1.11x} equals two, while $R$ approaches
$2/m_i^2$ as one moves out along the $i^{\mathrm{th}}$
cylindrical end.

A very similar analysis applies on the slices of constant time
in the Kastor-Traschen metrics~\cite{KastorTraschenBH},
solutions of the vacuum Einstein equations with positive
cosmological constant.

A simple example of a metric with two cylindrical ends with
\emph{toroidal transverse topology} is provided by Bianchi $I$
metrics in which two directions only have been compactified,
leading to a spatial topology $ \R \times {\Bbb T}^2$.

Another example of metrics with ends of cylindrical type is
provided by the extreme Kerr metrics.
The metric induced on Boyer-Lindquist sections of the event
horizons of the Kerr metric reads
\bel{15XII0.1}
 ds^2  = (R^2 + a^2 \cos^2 \theta  )\, d \theta^2 + \frac{ (R^2 + a^2 )^2 \sin^2 \theta }{ R^2 +
a^2 \cos^2 \theta } d \varphi^2
 \;,
\ee
where $R = m \pm \sqrt{m^2 -a^2}$. We note that its scalar curvature, denoted
here by $K$, is  \cite{KayllWWW}
$$
K =  \frac{(R^2 + a^2 ) (3 a^2 \cos^2 \theta  - R^2 )}{ (R^2 + a^2
 \cos^2 \theta  )^3}
 \;.
$$
We claim that the limiting metric, as one recedes to infinity
along the cylindrical end of the extreme Kerr metric, can be
obtained from \eq{15XII0.1} by setting $a=m$:
\bel{15XII0.2}
 \zg  = m^2 \bigg((1+ \cos^2 \theta  ) d \theta^2 + \frac{4 \sin^2 \theta }{1  +
 \cos^2 \theta } d \varphi^2\bigg)
 \;.
\ee
(This metric  has scalar curvature
$$
  K =  \frac {2 (3   \cos^2 \theta  - 1 )}{m^2 (1 +
 \cos^2 \theta  )^3}
 \;.
$$
and the reader should note that $K$ changes sign.)
Indeed, the \emph{extreme Kerr metrics} in Boyer-Lindquist
coordinates take the form, changing $\varphi$ to its negative
if necessary,
 \begin{eqnarray}
\label{Kerr2x} g & = & - dt^{2} + \frac{2  mr
}{r^{2}+m^{2}\cos^{2}\theta}(dt- m\sin^2 \theta d\varphi)^2 +  {(r^{2}+m^{2})
 \sin ^{2}\theta  d\varphi ^{2} }
 \nonumber
\\
& &+\frac{r^{2}+m^{2}\cos^{2}\theta}{(r-m)^2 }dr^{2}
+(r^{2}+m^{2}\cos^{2}\theta) d\theta ^{2}\;.
\end{eqnarray}
The metric induced on the slices $t=\const$ reads, keeping in
mind that $r>m$,
 \begin{eqnarray}
\label{Kerr}
 g & = &  \frac{r^{2}+m^{2}\cos^{2}\theta}{(r-m)^2}dr^{2}
 +(r^{2}+m^{2}\cos^{2}\theta)d\theta ^{2}
  \nonumber
\\
 &&
 + \frac{(r^{2}+m^{2})^{2}-(r-m)^2 m^{2}\sin ^{2}\theta
}{r^{2}+m^{2}\cos^{2}\theta}\sin ^{2}\theta  d\varphi ^{2}\;.
\end{eqnarray}
Introducing a new variable $x\in (-\infty,\infty)$ defined as
$$
 dx = -\frac{ dr}{r-m} \quad
 \Longrightarrow
 \quad
 x = -\ln {(r-m)}
 \;,
$$
so that $x$ tends to infinity as $r$ approaches $m$ from above,
the metric $g$ in \eq{Kerr} exponentially approaches
\begin{equation}
% g &\to_{x\to\infty}&   m^2(1+\cos^2 \theta) dx^2 + \zg \\  & = &
m^2 (1+\cos^2 \theta)\bigg(  dx^2 + d\theta^2 +  \frac{4 \sin^2 \theta }{(1  +
 \cos^2 \theta )^2} d \varphi^2\bigg)
%\underbrace{}_{=:\zg}\bigg)
\label{5I1.6}
\end{equation}
as $x \to \infty$. We thus see that the degenerate Kerr space-times contain CMC slices with
\emph{asymptotically conformally cylindrical} ends.

Recall that the scalar curvature $K$ of a metric of the form
$ d\theta^2 + e^{2f} d\varphi^2    $
equals
$$
% R_{\theta \varphi \theta}{}^{\varphi}
K= -2( f'' + (f')^2)
 \;.
$$
Hence the transverse part $\zg$ of the limiting conformal
metric appearing in \eq{5I1.6} has scalar curvature
$$
 K=-\frac{4 \cos (2 \theta )}{\left(\cos ^2\theta +1\right)^2}
   \;,
   $$
which is negative on the northern hemisphere and positive on
the southern one. We note that the slices $t=\const$ are
maximal, and the scalar constraint equation shows that $R\ge
0$. This example clearly exhibits the lack of correlation
between the sign of the limit $\lim_{x\to\infty}R$ and that of
the scalar curvature of the transverse part of the asymptotic
metric (whether $\zg$ as defined in \eq{15XII0.2} or its
conformally rescaled version $\zh$ from \eq{5I1.6}), even when
the constraint equations hold.

\bibliographystyle{amsplain}
\bibliography{../references/bartnik,%
../references/myGR,%
../references/newbiblio,%
../references/newbib,%
../references/reffile,%
../references/bibl,%
../references/Energy,%
../references/hip_bib,%
../references/dp-BAMS,%
../references/MLproposal,%
../references/prop,%
../references/besse2,%
../references/netbiblio%../dp-BAMS
}%

\def\polhk#1{\setbox0=\hbox{#1}{\ooalign{\hidewidth
  \lower1.5ex\hbox{`}\hidewidth\crcr\unhbox0}}} \def\cprime{$'$}
  \def\cprime{$'$} \def\cprime{$'$} \def\cprime{$'$} \def\cprime{$'$}
  \def\cprime{$'$}
\providecommand{\bysame}{\leavevmode\hbox to3em{\hrulefill}\thinspace}
\providecommand{\MR}{\relax\ifhmode\unskip\space\fi MR }
% \MRhref is called by the amsart/book/proc definition of \MR.
\providecommand{\MRhref}[2]{%
  \href{http://www.ams.org/mathscinet-getitem?mr=#1}{#2}
}
\providecommand{\href}[2]{#2}
\begin{thebibliography}{10}

\bibitem{Ak-Bot}
K.~Akutagawa and B.~Botvinnik, \emph{Yamabe metrics on cylindrical manifolds},
  Geom.\ Funct.\ Anal. \textbf{13} (2003), 259--333, arXiv:math/0107164.
  \MR{1982146 (2004e:53051)}

\bibitem{AkCarrMaz}
K~Akutagawa, G.~Carron, and R.~Mazzeo, \emph{The {Yamabe} problem on stratified
  spaces}, Geom. and Funct. Analysis (2012), in press, arXiv:1210.8054
  [math.DG].

\bibitem{AndChDiss}
L.~Andersson and P.T. Chru\'{s}ciel, \emph{On asymptotic behavior of solutions
  of the constraint equations in general relativity with ``hyperboloidal
  boundary conditions''}, Dissert. Math. \textbf{355} (1996), 1--100.
  \MR{MR1405962 (97e:58217)}

\bibitem{aviles:mcowen:noncompact:1988}
P.~Aviles and R.C. McOwen, \emph{Conformal deformation to constant negative
  scalar curvature on noncompact {Riemannian} manifolds}, Jour.\ Diff.\ Geom.
  \textbf{27} (1988), 225--239.

\bibitem{BaumgarteNaculich}
T.W. Baumgarte and S.G. Naculich, \emph{Analytical representation of a black
  hole puncture solution}, Phys.\ Rev.\ D \textbf{75} (2007), 067502, 4,
  arXiv:gr-qc/0701037. \MR{2312204 (2008a:83053)}

\bibitem{BeBeMaillot}
L.~Bessi{\`e}res, G.~Besson, and S.~Maillot, \emph{Ricci flow on open
  3-manifolds and positive scalar curvature}, Geom.\ Topol. \textbf{15} (2011),
  927--975, arXiv:1001.1458 [math.DG]. \MR{2821567}

\bibitem{Byde}
A.~Byde, \emph{Gluing theorems for constant scalar curvature manifolds},
  Indiana Univ.\ Math.\ Jour. \textbf{52} (2003), 1147--1199. \MR{MR2010322
  (2004h:53049)}

\bibitem{CaffarelliGidasSpruck}
L.A. Caffarelli, B.~Gidas, and J.~Spruck, \emph{Asymptotic symmetry and local
  behavior of semilinear elliptic equations with critical {S}obolev growth},
  Commun.\ Pure\ Appl.\ Math. \textbf{42} (1989), 271--297. \MR{MR982351
  (90c:35075)}

\bibitem{ChenLin95}
C.~C. Chen and C.~S. Lin, \emph{Local behavior of singular positive solutions
  of semilinear elliptic equations with {S}obolev exponent}, Duke\ Math.\ Jour.
  \textbf{78} (1995), 315--334. \MR{MR1333503 (96d:35035)}

\bibitem{ChIY}
Y.~Choquet-Bruhat, J.~Isenberg, and J.W. {York$,~$Jr.}, \emph{Einstein
  constraints on asymptotically {E}uclidean manifolds}, Phys.\ Rev.\ D
  \textbf{61} (2000), 084034 (20 pp.), arXiv:gr-qc/9906095.

\bibitem{christodoulou:murchadha}
D.~Christodoulou and N.{\'O} Murchadha, \emph{The boost problem in general
  relativity}, Commun.\ Math.\ Phys. \textbf{80} (1981), 271--300.

\bibitem{CMP}
P.T. Chru\'{s}ciel, R.~Mazzeo, and S.~Pocchiola, \emph{{Initial data sets with
  ends of cylindrical type: II. The vector constraint equation}}, Adv.\ Math.\
  and Theor.\ Phys. \textbf{17} (2013), 829--865, arXiv:1203.5138 [gr-qc].

\bibitem{CPP}
P.T. Chru\'{s}ciel, F.~Pacard, and D.~Pollack, \emph{{Singular {Yamabe metrics
  and initial data with exactly Kottler-Schwarzschild-de Sitter ends II.
  Generic} metrics}}, Math.\ Res.\ Lett. \textbf{16} (2009), 157--164,
  arXiv:0803.1817 [gr-qc]. \MR{2480569 (2009k:53079)}

\bibitem{ChPollack}
P.T. Chru\'{s}ciel and D.~Pollack, \emph{{Singular Yamabe metrics and initial
  data with \emph{exactly} Kottler--Schwarzschild--de Sitter ends}}, Ann.\
  Henri Poincar\'e \textbf{9} (2008), 639--654, arXiv:0710.3365 [gr-qc].
  \MR{2413198 (2009g:53051)}

\bibitem{GabachClement}
M.E.~Gabach Cl{\'e}ment, \emph{{Conformally flat black hole initial data, with
  one cylindrical end}}, Class.\ Quantum Grav. \textbf{27} (2010), 125010,
  arXiv:0911.0258 [gr-qc].

\bibitem{DGH}
M.~Dahl, R.~Gicquaud, and E.~Humbert, \emph{{A limit equation associated to the
  solvability of the vacuum Einstein constraint equations using the conformal
  method}},  (2010), arXiv:1012.2188 [gr-qc].

\bibitem{GabachDain1}
S.~Dain and M.E.~Gabach Cl{\'e}ment, \emph{Extreme {B}owen-{Y}ork initial
  data}, Class.\ Quantum Grav. \textbf{26} (2009), 035020, 16, arXiv:0806.2180
  [gr-qc]. \MR{2476223 (2010c:83007)}

\bibitem{GabachDain2}
\bysame, \emph{{Small deformations of extreme Kerr black hole initial data}},
  Class.\ Quantum Grav. \textbf{28} (2010), 075003, 20 pp., arXiv:1001.0178
  [gr-qc]. \MR{2777050 (2012a:83038)}

\bibitem{ErwannTT}
E.~Delay, \emph{Smooth compactly supported solutions of some underdetermined
  elliptic {PDE}, with gluing applications}, Commun.\ Partial Diff.\ Eq.
  \textbf{37} (2012), no.~10, 1689--1716, arXiv:1003.0535 [math.FA].
  \MR{2971203}

\bibitem{Estabrook}
F.~Estabrook, H.~Wahlquist, S.~Christensen, B.~DeWitt, L.~Smarr, and E.~Tsiang,
  \emph{{Maximally slicing a black hole}}, Phys.\ Rev. \textbf{D7} (1973),
  2814--2817.

\bibitem{FriedrichYamabe}
H.~Friedrich, \emph{Yamabe numbers and the {B}rill-{C}antor criterion}, Ann.\
  Henri Poincar\'e \textbf{12} (2011), 1019--1025. \MR{2802389}

\bibitem{GicquaudS}
R.~Gicquaud and A.~Sakovich, \emph{{A large class of non constant mean
  curvature solutions of the Einstein constraint equations on an asymptotically
  hyperbolic manifold}}, Commun.\ Math.\ Phys. \textbf{310} (2012), no.~3,
  705--763, arXiv:1012.2246 [gr-qc]. \MR{2891872}

\bibitem{GT}
D.~Gilbarg and N.S. Trudinger, \emph{Elliptic partial differential equations of
  second order}, Springer, Berlin, 1983.

\bibitem{HannamHusaNiall}
M.~Hannam, S.~Husa, and N.~\'O Murchadha, \emph{{Bowen-York trumpet data and
  black-hole simulations}}, Phys.\ Rev. \textbf{D80} (2009), 124007,
  arXiv:0908.1063 [gr-qc].

\bibitem{HannamEtAlPRL}
M.~Hannam, S.~Husa, D.~Pollney, B.~Br{\"u}gmann, and N.~{\'O}~Murchadha,
  \emph{Geometry and regularity of moving punctures}, Phys.\ Rev.\ Lett.
  \textbf{99} (2007), 241102, 4. \MR{2369068 (2008j:83001)}

\bibitem{Hebeyslides}
E.~Hebey, \emph{Existence, stability and instability for {Einstein--scalar
  field Lichnerowicz} equations}, 2009,
  \url{http://www.u-cergy.fr/rech/pages/hebey/IASBeamerFullPages.pdf}.

\bibitem{HPP}
E.~Hebey, F.~Pacard, and D.~Pollack, \emph{A variational analysis of
  {Einstein-scalar field Lichnerowicz equations on compact Riemannian}
  manifolds}, Commun.\ Math.\ Phys. \textbf{278} (2008), no.~1, 117--132,
  arXiv:gr-qc/0702031. \MR{MR2367200}

\bibitem{HNT3}
M.~Holst, G.~Nagy, and G.~Tsogtgerel, \emph{{Far-from-constant mean curvature
  solutions of Einstein's constraint equations with positive Yamabe metrics}},
  Phys.\ Rev.\ Lett. \textbf{100} (2008), no.~16, 161101, 4, arXiv:0802.1031
  [gr-qc]. \MR{2403263 (2009c:53112)}

\bibitem{Jimconstraints}
J.~Isenberg, \emph{Constant mean curvature solutions of the {Einstein}
  constraint equations on closed manifolds}, Class.\ Quantum Grav. \textbf{12}
  (1995), 2249--2274. \MR{MR1353772 (97a:83013)}

\bibitem{IMP2}
J.~Isenberg, R.~Mazzeo, and D.~Pollack, \emph{On the topology of vacuum
  spacetimes}, Ann.\ Henri Poincar\'e \textbf{4} (2003), 369--383.
  \MR{MR1985777 (2004h:53053)}

\bibitem{VinceJim:noncmc}
J.~Isenberg and V.~Moncrief, \emph{A set of nonconstant mean curvature
  solutions of the {E}instein constraint equations on closed manifolds},
  Class.\ Quantum Gravity \textbf{13} (1996), 1819--1847. \MR{MR1400943
  (97h:83010)}

\bibitem{KastorTraschenBH}
D.~Kastor and J.~Traschen, \emph{Cosmological multi-black-hole solutions},
  Phys.\ Rev.\ D (3) \textbf{47} (1993), 5370--5375, arXiv:hep-th/9212035.
  \MR{1225552 (94d:83046)}

\bibitem{Komatsu:2010fb}
E.~Komatsu et~al., \emph{{Seven-Year Wilkinson Microwave Anisotropy Probe
  (WMAP) Observations: Cosmological Interpretation}}, Astrophys.\ Jour.\ Suppl.
  \textbf{192} (2011), 18 (47 pp.), arXiv:1001.4538 [astr-ph.CO].

\bibitem{KMPS}
N.~Korevaar, R.~Mazzeo, F.~Pacard, and R.~Schoen, \emph{Refined asymptotics for
  constant scalar curvature metrics with isolated singularities}, Invent.\
  Math. \textbf{135} (1999), 233--272. \MR{MR1666838 (2001a:35055)}

\bibitem{KayllWWW}
K.~Lake, \url{http://grtensor.org/blackhole}.

\bibitem{NiallEdwardCMCSchwarzschild}
E.~Malec and N.~{\'O} Murchadha, \emph{Constant mean curvature slices in the
  extended {S}chwarzschild solution and the collapse of the lapse}, Phys.\
  Rev.\ D (3) \textbf{68} (2003), no.~12, 124019, 16 pp., arXiv:gr-qc/0307046.
  \MR{2071735 (2005f:83017)}

\bibitem{Marques}
F.C. Marques, \emph{Isolated singularities of solutions of the {Yamabe}
  equation}, Calc. of Var. \textbf{32} (2008), 349--371,
  doi:10.1007/s00526-007-0144-3.

\bibitem{MaxwellNonCMC}
D.~Maxwell, \emph{{A class of solutions of the vacuum Einstein constraint
  equations with freely specified mean curvature}}, Math.\ Res.\ Lett.
  \textbf{16} (2008), 627--645, arXiv:0804.0874 [gr-qc]. \MR{2525029
  (2010j:53057)}

\bibitem{MPa}
R.~Mazzeo and F.~Pacard, \emph{Constant scalar curvature metrics with isolated
  singularities}, Duke Math.\ Jour. \textbf{99} (1999), 353--418. \MR{MR1712628
  (2000g:53035)}

\bibitem{MPo}
R.~Mazzeo and D.~Pollack, \emph{Gluing and moduli for noncompact geometric
  problems}, Geometric theory of singular phenomena in partial differential
  equations (Cortona, 1995), Sympos.\ Math., XXXVIII, Cambridge Univ. Press,
  Cambridge, 1998, pp.~17--51. \MR{MR1702086 (2000i:53058)}

\bibitem{MPU2}
R.~Mazzeo, D.~Pollack, and K.~Uhlenbeck, \emph{Connected sum constructions for
  constant scalar curvature metrics}, Topol.\ Methods Nonlinear Anal.
  \textbf{6} (1995), 207--233. \MR{MR1399537 (97e:53076)}

\bibitem{MPU1}
\bysame, \emph{Moduli spaces of singular {Y}amabe metrics}, Jour.\ Amer.\
  Math.\ Soc. \textbf{9} (1996), 303--344. \MR{MR1356375 (96f:53055)}

\bibitem{P93}
D.~Pollack, \emph{Compactness results for complete metrics of constant positive
  scalar curvature on subdomains of {$S\sp n$}}, Indiana Univ.\ Math.\ Jour.
  \textbf{42} (1993), 1441--1456. \MR{MR1266101 (95c:53052)}

\bibitem{Rat}
J.~Ratzkin, \emph{An end to end gluing construction for metrics of constant
  positive scalar curvature}, Indiana Univ.\ Math.\ Jour. \textbf{52} (2003),
  703--726. \MR{MR1986894 (2004m:53066)}

\bibitem{Riess:2006fw}
{A.G.} Riess~\emph{et al.}, \emph{New {Hubble Space Telescope} discoveries of
  type {Ia Supernovae at $z > 1$: N}arrowing constraints on the early behavior
  of dark energy}, Astroph.\ Jour. \textbf{659} (2007), 98--121,
  arXiv:astro-ph/0611572.

\bibitem{SchoenSingular}
R.~Schoen, \emph{The existence of weak solutions with prescribed singular
  behavior for a conformally invariant scalar equation}, Commun.\ Pure Appl.\
  Math. \textbf{41} (1988), 317--392. \MR{MR929283 (89e:58119)}

\bibitem{Stanciulescu}
C.~Stanciulescu, \emph{Spherically symmetric solutions of the vacuum {E}instein
  field equations with positive cosmological constant}, 1998, Diploma Thesis,
  University of Vienna.

\bibitem{Sullivan2}
D.~Sullivan, \emph{Related aspects of positivity in {R}iemannian geometry},
  Jour.\ Diff.\ Geom. \textbf{25} (1987), 327--351. \MR{882827 (88d:58132)}

\bibitem{Waxenegger}
G.~Waxenegger, \emph{Black hole initial data with one cylindrical end}, Ph.D.
  thesis, University of Vienna.

\bibitem{BMW}
G.~Waxenegger, R.~Beig, and N.\'O Murchadha, \emph{{Existence and uniqueness of
  Bowen-York Trumpets}}, Class.\ Quantum Grav. \textbf{28} (2011), 245002,
  pp.~15, arXiv:1107.3083 [gr-qc]. \MR{2865319}

\bibitem{WoodVasey:2007jb}
W.M. Wood-Vasey et~al., \emph{{Observational Constraints on the Nature of the
  Dark Energy: First Cosmological Results from the ESSENCE Supernova Survey}},
  Astroph.\ Jour. \textbf{666} (2007), 694--715, arXiv:astro-ph/0701041.

\end{thebibliography}
\end{document}